\documentclass[sigconf]{acmart}
\usepackage{bbm}
\usepackage{bm}
\usepackage{epsfig,endnotes}
\usepackage{multirow}
\usepackage{subfig}
\usepackage{graphicx}
\usepackage{grffile}

\newcounter{subcopyrightbox@save}
\usepackage[font=bf]{caption}
\usepackage{color, url}
\usepackage{xspace} 

\usepackage{amsmath}

\usepackage{mathrsfs}
\usepackage{amssymb}
\usepackage{amsmath}
\usepackage{amsthm}
\usepackage{epstopdf}
\usepackage{balance}

\usepackage{thm-restate,thmtools}

\usepackage[ruled,linesnumbered,noend]{algorithm2e}

\newcommand{\myparatight}[1]{\smallskip\noindent{\bf {#1}:}~}

\usepackage{algorithmic}

\newcommand{\argmax}{\operatornamewithlimits{argmax}}
\newcommand{\argmin}{\operatornamewithlimits{argmin}}

\allowdisplaybreaks

\AtBeginDocument{%
  \providecommand\BibTeX{{%
    \normalfont B\kern-0.5em{\scshape i\kern-0.25em b}\kern-0.8em\TeX}}}

\setcopyright{acmcopyright}
\copyrightyear{2020}
\acmYear{2020}
\acmDOI{10.1145/xxxxxxx.xxxxxxx}

\setcopyright{acmcopyright}

\copyrightyear{2020}
\acmYear{2020}
\setcopyright{iw3c2w3}
\acmConference[WWW '20]{Proceedings of The Web Conference 2020}{April 20--24, 2020}{Taipei, Taiwan}
\acmBooktitle{Proceedings of The Web Conference 2020 (WWW '20), April 20--24, 2020, Taipei, Taiwan}
\acmPrice{}
\acmDOI{10.1145/3366423.3380029}
\acmISBN{978-1-4503-7023-3/20/04}

\begin{document}


\title{Certified Robustness of Community Detection against Adversarial Structural Perturbation via Randomized Smoothing}

\author{Jinyuan Jia, Binghui Wang, Xiaoyu Cao, Neil Zhenqiang Gong}
\thanks{The first two authors made equal contributions.}
\affiliation{%
  \institution{Duke University}
}
\email{{jinyuan.jia,binghui.wang,xiaoyu.cao, neil.gong}@duke.edu}

\keywords{Community Detection; Certified Robustness}


\begin{abstract}
 Community detection plays a key role in understanding graph structure. However, several recent studies showed that  community detection is vulnerable to adversarial structural perturbation. In particular, via adding or removing a small number of carefully selected edges in a graph, an attacker can manipulate the detected communities. However, to the best of our knowledge, there are no studies on certifying robustness of community detection against such adversarial structural perturbation. In this work, we aim to bridge this gap. 
 Specifically,  we develop the first certified robustness guarantee of community detection against adversarial structural perturbation. 
 Given an arbitrary community detection method, we build a new \emph{smoothed} community detection method via randomly perturbing the graph structure. We theoretically show that the smoothed community detection method provably groups a given arbitrary set of nodes into the same community (or different communities) when the number of edges added/removed by an attacker is bounded. Moreover, we show that our certified robustness is \emph{tight}. We also empirically evaluate our method on multiple real-world graphs with ground truth communities. 
 
\end{abstract}
\maketitle

\section{Introduction}
Graph is a powerful tool to represent many complex systems. 
For example, an online social network can be viewed as a graph, where nodes are users and edges represent friendships or interactions between users. 
Community detection is a basic tool to understand the structure of a graph and has many applications. For instance, communities in a social graph may represent users with common interests, locations, occupations, etc.. 
Therefore, many community detection methods (e.g.,~\cite{blondel2008fast,girvan2002community,newman2006modularity,fortunato2007resolution,leskovec2009community,leskovec2008statistical,yang2015defining,leskovec2010empirical}) have been proposed by various fields such as network science, applied physics, and bioinformatics. Roughly speaking, a community detection method divides the nodes in a graph into groups such that nodes in the same group are densely connected and nodes in different groups are sparsely connected. 

However, multiple recent studies showed that  community detection  is vulnerable to adversarial structural perturbations~\cite{nagaraja2010impact,waniek2018hiding,fionda2017community,chen2019ga,chen2017practical}. 
Specifically, via adding or removing a small number of carefully selected edges in a graph, an attacker can manipulate the detected communities. 
For example, an attacker can spoof a community detection method to split a set of nodes, which are originally detected as in the same community, into different communities. An attacker can also spoof a community detection method to merge a set of nodes, which are originally detected as in different communities, into the same community. We call these two attacks \emph{splitting attack} and \emph{merging attack}, respectively. However, to the best of our knowledge, there are no studies to certify robustness of community detection against such adversarial structural perturbation. We note that several heuristic defenses~\cite{nagaraja2010impact,chen2017practical} were proposed to enhance the robustness of community detection against structural perturbation. However, these defenses lack formal guarantees and can often be defeated by strategic attacks that adapt to them. 

In this work, we aim to bridge this gap. In particular, we aim to develop certified robustness of community detection against structural perturbation. Given an arbitrary community detection method, our techniques transform the method to a robust community detection method that provably groups a given arbitrary set of nodes into the same community (against splitting attacks) or into different communities (against merging attacks) when the number of edges added/removed by the attacker is no larger than a threshold. We call the threshold \emph{certified perturbation size}.

Our robustness guarantees are based on a recently proposed technique called \emph{randomized smoothing}~\cite{lecuyer2018certified,li2018second,cohen2019certified}, which is the state-of-the-art method to build provably robust machine learning methods. Specifically, given an arbitrary function $f$, which takes $\mathbf{x}$ as an input and outputs a categorial value. Randomized smoothing constructs a smoothed function $g$ via adding random noise to the input $\mathbf{x}$. Moreover, the output of the smoothed function is the function $f$'s output that has the largest probability when adding random noise to the input $\mathbf{x}$. Suppose an attacker can add a perturbation to the input $\mathbf{x}$. The smoothed function provably has the same output once the perturbation added to the input $\mathbf{x}$ is bounded.

We propose to certify robustness of community detection using randomized smoothing. Specifically, given a graph, an arbitrary community detection method, and an arbitrarily set of nodes in the graph, we construct a function $f$, which takes the graph as an input and outputs 1 if the community detection method groups the set of nodes  into the same community, otherwise the function $f$ outputs 0. Then, we build a smoothed function $g$ via adding random noise to the graph structure, i.e., randomly adding or removing edges in the graph. Finally, we certify the robustness of the smoothed function $g$ against adversarial structural perturbation. 

However, existing randomized smoothing methods are insufficient to certify robustness of community detection. Specifically, they assume the input $\mathbf{x}$ is continuous and add Gaussian or Laplacian noise to it. 
However, graph structure is  {binary} data, i.e., a pair of nodes can be connected or unconnected; and Gaussian or Laplacian noise is not semantically meaningful for such binary data. To address the challenge, we develop randomized smoothing for binary data.  We theoretically derive a certified perturbation size via addressing several technical challenges. For instance, we prove a variant of the Neyman-Pearson Lemma~\cite{neyman1933ix} for binary data; and we divide the graph structure space into regions in a novel way such that we can apply the  variant of the Neyman-Pearson Lemma to certify robustness of community detection. Moreover, we prove that our  certified perturbation size is \emph{tight} if no assumptions on the community detection method are made. Our certified perturbation size is the solution to an optimization problem.  Therefore, we further design an algorithm to solve the optimization problem. 

We empirically evaluate our method using multiple real-world graph datasets with ground-truth communities including Email, DBLP, and Amazon datasets. We choose the efficient community detection method called Louvain's method~\cite{blondel2008fast}. We study the impact of various parameters on the certified robustness.

In summary, our key contributions are as follows: 

\begin{itemize}
\item We develop the first certified robustness of community detection against adversarial structural perturbation.  Moreover, we show that our certified robustness is tight. 
\item Our certified perturbation size is the solution to an optimization problem  and we develop an algorithm to solve the optimization problem. 
\item We evaluate our method on multiple real-world datasets.  
\end{itemize}

\section{Background}
\subsection{Community Detection}
Suppose we are given an undirected graph $G=(V,E)$, where $V$ is the set of nodes and $E$ is the set of edges. 
A community detection method divides the nodes in the graph into  groups, which are called communities. In non-overlapping community detection, a node only belongs to one community, while in overlapping community detection, a node may belong to multiple communities. 
Formally, a community detection algorithm  $\mathcal{A}$ takes a graph as an input and produces a set of communities $\mathcal{C}=\{C_1,C_2,\cdots,C_k\}$, where $V=\cup_{i=1}^{k}C_i$ and $C_i$ is the set of nodes that are in the $i$th community. For simplicity, we represent the graph structure as a binary vector $\mathbf{x}$, where an entry of the vector represents the connection status of the corresponding pair of nodes. Specifically, an entry $x_i=1$ if the corresponding pair of nodes are connected, otherwise $x_i=0$. Moreover, we denote by $n$ the length of the binary vector $\mathbf{x}$. Therefore, we can represent community detection  as $\mathcal{C}=\mathcal{A}(\mathbf{x})$.

\subsection{Attacks to Community Detection}

\myparatight{Adversarial structural perturbation} We consider an attacker can manipulate the graph structure, i.e., adding or removing some edges in the graph. In particular,  an attacker may have control of some nodes in the graph and can add or remove edges among them.  For instance, in a social graph, the attacker-controlled nodes may be fake users created by the attacker or normal users compromised by the attacker. We denote by a binary vector $\bm{\delta}$ the attacker's perturbation to the graph, where $\delta_i=1$ if and only if the attacker changes the connection status of the corresponding pair of nodes. $\mathbf{x}\oplus \bm{\delta}$ is the perturbed graph structure, where the operator $\oplus$ is the XOR between two binary variables. Moreover, we use $||\bm{\delta}||_0$ to measure the magnitude of the perturbation because $||\bm{\delta}||_0$ has semantic interpretations. In particular,  $||\bm{\delta}||_0$ is the number of edges added or removed by the attacker.

\myparatight{Two attacks}  An attacker can manipulate the detected communities via adversarial structural perturbation~\cite{nagaraja2010impact,waniek2018hiding,fionda2017community,chen2019ga,chen2017practical}. Specifically, there are two types of attacks  to community detection: 
\begin{itemize}
\item {\bf Splitting attack.} Given a set of nodes (called \emph{victim nodes}) $\Gamma=\{u_1,u_2,\cdots,u_c\}$ that are in the same community. A splitting attack aims to perturb the graph structure such that a community detection method divides the nodes in $\Gamma$  into different communities.  Formally, we have communities $\mathcal{C}'=\{C_1,C_2,\cdots,C_{k'}\}=\mathcal{A}(\mathbf{x}\oplus \bm{\delta}')$ after the attacker adds perturbation $\bm{\delta}'$ to the graph structure, but there does not exist a community $C_i$ such that $\Gamma\subset C_i$.

\item {\bf Merging attack.} Given a set of victim nodes $\Gamma$ that are in different communities. A merging attack aims to perturb the graph structure such that a community detection method groups the nodes in $\Gamma$  into the same community.  Formally, we have communities $\mathcal{C}''=\{C_1,C_2,\cdots,C_{k''}\}=\mathcal{A}(\mathbf{x}\oplus \bm{\delta}'')$ after the attacker adds perturbation $\bm{\delta}''$ to the graph structure, and  there exists a community $C_i$ such that $\Gamma\subset C_i$. 

\end{itemize}

We aim to develop certified robustness of community detection against the splitting and merging attacks.

\subsection{Randomized Smoothing}
Randomized smoothing is state-of-the-art method to build provably secure machine learning methods~\cite{cao2017mitigating,cohen2019certified}.  Suppose we are given a function $f$, which takes $\hat{\mathbf{x}}$ as an input and outputs  a categorical value in a domain $\{1,2,\cdots,d\}$. Randomized smoothing aims to construct a smoothed function $g$ via adding random noise $\hat{\bm{\epsilon}}$ to the input $\hat{\mathbf{x}}$. Moreover, the output of the smoothed function is the output of the function $f$ that has the largest probability when adding random noise to the input $\hat{\mathbf{x}}$. Formally, we have:
\begin{align}
g(\hat{\mathbf{x}})=\argmax_{\hat{y}\in \{1,2,\cdots,d\}}\text{Pr}(f(\hat{\mathbf{x}}+\hat{\bm{\epsilon}})=\hat{y}),
\end{align}
where $\hat{\bm{\epsilon}}$ is random noise drawn from a certain distribution. Suppose an attacker can add perturbation $\hat{\bm{\delta}}$ to the input $\hat{\mathbf{x}}$. Existing studies~\cite{lecuyer2018certified,li2018second,cohen2019certified} assume $\hat{\mathbf{x}}$ is continuous data. Moreover, they showed that, when the random noise is drawn from a Gaussian distribution or Laplacian distribution, the smoothed function provably has the same output when the $L_2$-norm or $L_1$-norm of the perturbation   $\hat{\bm{\delta}}$ is bounded. However, in our problem, the graph structure is binary data. Gaussian or Laplacian noise is not semantically meaningful for such binary data. To address the challenge, we will develop randomized smoothing for binary data and apply it to  certify robustness against the splitting and merging attacks.

\section{Certified Robustness}
\subsection{Randomized Smoothing on Binary Data} 
We first construct a function $f$ to model the splitting and merging attacks. Specifically, given a graph whose structure we represent as a binary vector $\mathbf{x}$, a community detection algorithm $\mathcal{A}$, and a set of victim nodes denoted as $\Gamma$, the function $f$ outputs 1 if the nodes in $\Gamma$ are grouped into the same community detected by $\mathcal{A}$ and outputs 0 otherwise. Formally, we define $f$ as follows:
\begin{align}
\label{definition_of_classifier_f}
f(\mathbf{x}) = \begin{cases}
    &1,  \text{ if }\exists i, \Gamma\subset C_i, { where }\ C_i \in \mathcal{A}(\mathbf{x}) \\ 
    & 0, \text{ otherwise.}
    \end{cases}
\end{align}
We simply use $\mathbf{x}$ as an input for the function $f$ because we study structural perturbation and other parameters--such as the community detection algorithm $\mathcal{A}$ and the set of victim nodes $\Gamma$--can be assumed to be constants.  
An attacker adds a perturbation vector $\bm{\delta}$ to the graph structure, i.e., $\mathbf{x} \oplus \bm{\delta}$ is the perturbed structure.  
When the nodes in $\Gamma$ are in the same community before attack (i.e., $f(\mathbf{x})=1$) and $f(\mathbf{x}\oplus \bm{\delta})$ produces 0, a splitting attack succeeds. When the nodes in $\Gamma$ are in different communities before attack (i.e., $f(\mathbf{x})=0$) and $f(\mathbf{x}\oplus \bm{\delta})$ produces 1, a merging attack succeeds.

We construct a smoothed function $g$ via adding random noise to the graph structure $\mathbf{x}$. Specifically,  
we define a  noise distribution in the discrete space $\{0,1\}^{n}$ as follows: 
\begin{align}
\label{define_of_epsilon_variable}
    \text{Pr}(\epsilon_i=0)=\beta, \text{Pr}(\epsilon_i=1)= 1 -\beta, \forall i \in \{1,2,\cdots,n\},
\end{align}
where $n$ is the length of $\mathbf{x}$ and $\epsilon_i$ is the random binary noise added to the $i$th entry of $\mathbf{x}$. Formally, $\mathbf{x}\oplus \bm\epsilon$ is the noisy graph structure. Our random noise means that the connection status (connected or unconnected) of a pair of nodes is preserved with a probability $\beta$ and changed with a probability $1-\beta$.

We note that the detected communities $\mathcal{C}=\mathcal{A}(\mathbf{x}\oplus {\bm\epsilon})$ are random since ${\bm\epsilon}$ is random. Therefore, the output $f(\mathbf{x}\oplus \bm{\epsilon})$ is also random. The smoothed function $g$ outputs the value that has a larger probability. Formally, we have:
\begin{align}
\label{define_of_smoothed_classifier}
    g(\mathbf{x})&=\argmax_{y\in \{0,1\}} \text{Pr}(f(\mathbf{x}\oplus\bm{\epsilon})=y)\nonumber \\
    &=
    \begin{cases}
        & 1, \text{ if }\text{Pr}(f(\mathbf{x}\oplus\bm{\epsilon})=1)> 0.5 \\
        & 0, \text{ otherwise.}
    \end{cases}
\end{align}

Certifying robustness against a splitting attack is to certify that $g(\mathbf{x}\oplus\bm{\delta})=1$ for all $||\bm{\delta}||_0\leq L_1$, while certifying robustness against a merging attack is to certify that $g(\mathbf{x}\oplus\bm{\delta})=0$ for all $||\bm{\delta}||_0\leq L_2$. In other words, we aim to certify that $g(\mathbf{x}\oplus\bm{\delta})=y$ for all $||\bm{\delta}||_0\leq L$, where $y\in \{0, 1\}$ and $L$ is called \emph{certified perturbation size}.

\subsection{Deriving Certified Perturbation Size}
In this section, we derive the certified perturbation size of the smoothed function $g$ theoretically for a given graph, community detection algorithm, and a set of victim nodes. In the next section, we will design algorithms to compute the certified perturbation size in practice. Our results can be summarized in the following two theorems.

\begin{restatable}[Certified Perturbation Size]{thm}{certyfiedperturbationsize}
\label{CertifiedPerturbationSize}
Given a graph-structure binary vector $\mathbf{x}$, a community detection algorithm $\mathcal{A}$, and a set of victim nodes $\Gamma$.  
The function $f$, random noise $\bm{\epsilon}$, and smoothed function $g$ are defined in  Equation~\ref{definition_of_classifier_f},~\ref{define_of_epsilon_variable}, and~\ref{define_of_smoothed_classifier}, respectively.  Assume there exists $\underline{p}\in [0,1]$ such that: 
\begin{align}
\label{main_theorem_condition_label}
\text{Pr}(f(\mathbf{x} \oplus \bm{\epsilon})=y)\geq \underline{p} > 0.5,
\end{align}
where $\underline{p}$ is a lower bound of the probability $p=\text{Pr}(f(\mathbf{x} \oplus \bm{\epsilon})=y)$  that $f$ outputs $y$ under the random noise $\bm{\epsilon}$. 
Then, we have: 
\begin{align}
g(\mathbf{x} \oplus \bm{\delta})=y, \forall   ||\bm{\delta}||_0 \leq L, 
\end{align}
where $L$ is called certified perturbation size and is the solution to the following optimization problem: 
\begin{align}
\label{problemK}
&L= \argmax l,   \\
& \textrm{s.t. } ||\bm{\delta}||_0 = l, \\
&\quad \sum_{i=1}^{\mu-1} \text{Pr}(\mathbf{x} \oplus \bm{\delta} \oplus \bm{\epsilon} \in \mathcal{H}_i) \nonumber  \\ 
\label{inequalityconstraint}
& \quad + (\underline{p} - \sum_{i=1}^{\mu-1} \text{Pr}(\mathbf{x} \oplus \bm{\epsilon} \in \mathcal{H}_i)) \cdot \frac{\text{Pr}(\mathbf{x} \oplus \bm{\delta} \oplus \bm{\epsilon} \in \mathcal{H}_{\mu})}{\text{Pr}(\mathbf{x} \oplus \bm{\epsilon} \in \mathcal{H}_{\mu})} 
 >0.5, 
\end{align}
where we define region $\mathcal{H}(e) = \{\mathbf{z}\in \{0,1\}^n: \frac{\text{Pr}(\mathbf{x} \oplus \bm{\epsilon} = \mathbf{z})}{\text{Pr}(\mathbf{x} \oplus \bm{\delta} \oplus \bm{\epsilon} = \mathbf{z})} = \big(\frac{\beta}{1-\beta}\big)^{e}\}$ and density ratio $h(e)=\big(\frac{\beta}{1-\beta}\big)^{e}$, where $e=-n, -n+1, \cdots, n-1, n$. We rank the regions $\mathcal{H}(-n), \mathcal{H}(-n+1), \cdots, \mathcal{H}(n)$ in a descending order with respect to the density ratios $h(-n), h(-n+1), \cdots, h(n)$. Moreover, we denote the ranked regions as $\mathcal{H}_1, \mathcal{H}_2, \cdots, \\ \mathcal{H}_{2n+1}$. Furthermore, $\mu$ is defined as follows:
$$\mu = \argmin_{\mu^{\prime} \in \{1, 2, \cdots, 2n+1\}} \mu^{\prime}, \ s.t. \ \sum_{i=1}^{\mu^{\prime}} \textrm{Pr}(\mathbf{x} \oplus \bm{\epsilon} \in \mathcal{H}_i) \geq \underline{p}$$
\end{restatable}
\begin{proof}
    See Appendix~\ref{proof_of_theorem_certifiedperturbationsize}. 
\end{proof}

Next, we show that our certified perturbation size is tight.  

\begin{restatable}[Tightness of the Certified Perturbation Size]{thm}{tightnessofcertifiedperturbationsize}
\label{tighnessbound}
For any perturbation $\bm{\delta}$ with $||\bm{\delta}||_0 > L$, there exists a community detection algorithm $\mathcal{A}^{*}$ (and thus a function $f^*$) consistent with Equation~\ref{main_theorem_condition_label} such that $g(\mathbf{x} \oplus \bm{\delta})\neq y$ or there exists ties.  
\end{restatable}
\begin{proof}
    See Appendix~\ref{proof_of_tight_of_certified_perturbation_size}. 
\end{proof}
We have the following observations from our two theorems: 
\begin{itemize}
\item Our certified perturbation size can be applied to any community detection method.  
\item Our certified perturbation size depends on $\underline{p}$ and $\beta$. When the probability lower bound $\underline{p}$ is tighter, our certified perturbation size is larger. We use the probability lower bound $\underline{p}$ instead of its exact value $p$ because it is challenging to compute the exact value.
\item When using the noise distribution defined in Equation~\ref{define_of_epsilon_variable} and  no further assumptions are made on the community detection algorithm, it is impossible to certify a perturbation size that is larger than $L$. 
\end{itemize}

\subsection{Computing Certified Perturbation Size}
Given a graph-structure binary vector $\mathbf{x}$, a community detection algorithm $\mathcal{A}$, and a set of victim nodes $\Gamma$, we aim to compute the certified perturbation size in practice. We face two challenges. 
The first challenge is to estimate $y$ and obtain the probability lower bound $\underline{p}$. The second challenge is how to solve the optimization problem  in Equation~\ref{problemK}. To address the first challenge, we first estimate a value of $y$, and then use the one-sided Clopper-Pearson method~\cite{brown2001interval} to estimate the probability bound with probabilistic guarantees. To address the second challenge, we 
develop an efficient algorithm to solve the optimization problem. 

\myparatight{Estimating $y$ and $\underline{p}$} 
We leverage a Monte-Carlo method to estimate $y$ and $\underline{p}$ with probabilistic guarantees. Specifically, we first randomly sample $N$ noise, and we use $\bm{\epsilon}_1,\bm{\epsilon}_2,\cdots,\bm{\epsilon}_N$ to denote them. Then, we  compute the frequency of the output 0 and 1 for the function $f$, i.e., $m_{0}=\sum_{i=1}^{N}\mathbb{I}(f(\mathbf{x}\oplus\bm{\epsilon}_i)=0)$ and $m_{1}=\sum_{i=1}^{N}\mathbb{I}(f(\mathbf{x}\oplus\bm{\epsilon}_i)=1)$, where $\mathbb{I}$ is an indicator function. We estimate $\hat{y}=\argmax_{i\in\{0,1\}}m_{i}$. Then, we estimate $\underline{p}$ by leveraging the one-sided Clopper-Pearson method. 
Estimating $\underline{p}$  can be viewed as estimating the parameter of a Binomial distribution. In particular,  $m_{\hat{y}}$ can be viewed as a sample from a Binomial distribution  $Bin(N,p)$, where $m_{\hat{y}}$ is the frequency of the value $\hat{y}$ and $Bin(N,p)$ denotes a Binomial distribution with parameters $N$ and $p$. Therefore, we can estimate $\underline{p}$ by leveraging the one-sided Clopper-Pearson method. Specifically, we have: 
\begin{align}
    \label{equation_to_estimate_pa}
    \underline{p}=B(\alpha;m_{\hat{y}},N-m_{\hat{y}}+1),
\end{align}
where $1-\alpha$ represents the confidence level and $B(\alpha;m_{\hat{y}},N-m_{\hat{y}}+1)$ denotes the $\alpha$th quantile of the beta distribution with parameters $m_{\hat{y}}$ and $N-m_{\hat{y}}+1$.

\myparatight{Solving the optimization problem} After obtaining the probability bound $\underline{p}$, we solve the optimization problem in Equation~\ref{problemK} to obtain $L$. The key to solve the optimization problem is to compute $\text{Pr}(\mathbf{x} \oplus \bm{\epsilon} \in \mathcal{H}(e))$ and $\text{Pr}(\mathbf{x} \oplus \bm{\delta} \oplus \bm{\epsilon} \in \mathcal{H}(e))$ for each $e\in \{-n, -n+1, \cdots, n\}$ when $||\bm{\delta}||_0=l$. Specifically, we have:
\begin{align}
\label{probabilityX}
&\text{Pr}(\mathbf{x} \oplus \bm{\epsilon} \in \mathcal{H}(e)) 
= \sum_{i=\max\{0,e\}}^{\min\{n,n+e\}} \beta^{n-(i-e)} (1-\beta)^{(i-e)} \cdot \theta(e,i)\\
\label{probabilityY}
&\text{Pr}(\mathbf{x} \oplus \bm{\delta} \oplus \bm{\epsilon} \in \mathcal{H}(e)) 
=\sum_{i=\max\{0,e\}}^{\min\{n,n+e\}} \beta^{n-i} (1-\beta)^i \cdot \theta(e,i),
\end{align}
where $\theta(e,i)$ is defined as follows:
{\footnotesize
\begin{align}
\theta(e,i) = 
    \begin{cases}
      0,  &\textrm{ if } (e+l) \textrm{ mod } 2 \neq 0, \\
      0,  &\textrm{ if } 2i-e<l, \\
  {n-l \choose \frac{2i-e-l}{2}} {l \choose \frac{l-e}{2}}, &\textrm{ otherwise } 
  \end{cases}
\end{align}
}
The calculation details can be found in Appendix~\ref{compute_algorithm_part}. Once we can compute the probabilities $\text{Pr}(\mathbf{x} \oplus \bm{\epsilon} \in \mathcal{H}(e))$ and $\text{Pr}(\mathbf{x} \oplus \bm{\delta} \oplus \bm{\epsilon} \in \mathcal{H}(e))$, we can iteratively find the largest $l$ such that the constraint in Equation~\ref{inequalityconstraint} is satisfied. Such largest $l$ is our certified perturbation size $L$.

\myparatight{Complete certification algorithm}  Algorithm~\ref{alg:certify} shows our complete certification algorithm. 
The function \textsc{SampleUnderNoise}  randomly samples $N$ noise from the noise distribution defined in Equation~\ref{define_of_epsilon_variable}, adds each noise to the graph structure, and computes the frequency of the function $f$'s output 0 and 1. Then, our algorithm estimates $\hat{y}$ and $\underline{p}$. Based on $\underline{p}$, the function \textsc{CertifiedPerturbationSize} computes the certified perturbation size by solving the optimization problem in Equation~\ref{problemK}. 
Our algorithm returns ($\hat{y}$, $L$) if $\underline{p} >0.5$ and ABSTAIN otherwise. The following proposition shows the probabilistic guarantee of our certification algorithm.

\begin{algorithm}[t]
    \SetAlgoLined
    \SetNoFillComment
    \DontPrintSemicolon
    \KwIn{$f$, $\beta$, $\mathbf{x}$, $N$, $\alpha$.}
    \KwOut{ABSTAIN or $(\hat{y},L)$.}
    $m_0$, $m_1=  \textsc{SampleUnderNoise}(f,\beta,\mathbf{x},N)$ \\
    $\hat{y}=\argmax_{i\in\{0,1\}}m_i$ \\
    $\underline{p}=B(\alpha;m_{\hat{y}},N-m_{\hat{y}}+1)$ \\
    
    \If{$\underline{p} >0.5$}{
        $L=\textsc{CertifiedPerturbationSize}(\underline{p})$\\
        \Return{$(\hat{y},L)$} 
    }
    \Else{
        \Return{\textup{ABSTAIN}}
    }
    \caption{\textsc{Certify}}
    \label{alg:certify}
\end{algorithm}

\begin{restatable}{propost}{probabilisticguarantee}
    \label{probability_guarantees}
    With probability at least $1-\alpha$ over the randomness in Algorithm~\ref{alg:certify}, if the algorithm returns an output value $\hat{y}$ and a certified perturbation size $L$ (i.e., does not ABSTAIN), then we have $g(\mathbf{x}\oplus\bm{\delta})=\hat{y}, \forall ||\bm{\delta}||_0 \leq L$.
\end{restatable}
\begin{proof}
    See Appendix~\ref{proof_of_probability_guarantee}. 
\end{proof}

\section{Evaluation}

\begin{table}[!t] 
\centering
\caption{Dataset statistics.}
\centering
\begin{tabular}{|c|c|c|c|c|} \hline 
{\bf Dataset} & {\small \#Nodes} & {\small \#Edges}  & {\small \#Communities}   \\ \hline
{\bf Email} &  {\small 1,005} & {\small 25,571} & {\small 42} \\ \hline
{\bf DBLP} &  {\small 317,080} & {\small 1,049,866} & {\small 13,477} \\ \hline
{\bf Amazon} &  {\small 334,863} & {\small 925,872} & {\small 75,149} \\ \hline
\end{tabular}
\label{dataset_stat}
\end{table}

\begin{figure*}[!h]
\center
\subfloat[Email]{\includegraphics[width=0.32\textwidth]{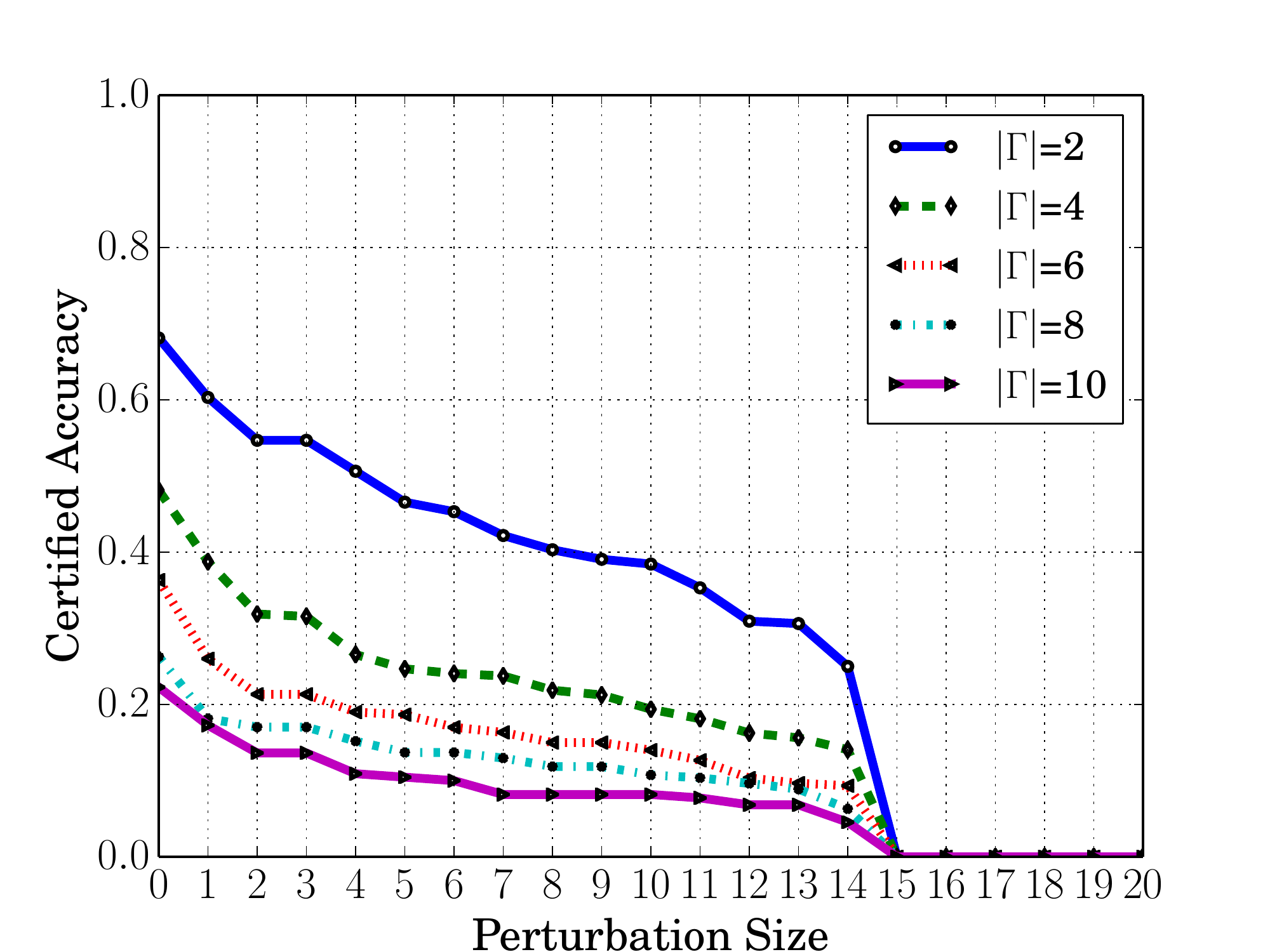} \label{fig1}} 
\subfloat[DBLP]{\includegraphics[width=0.32\textwidth]{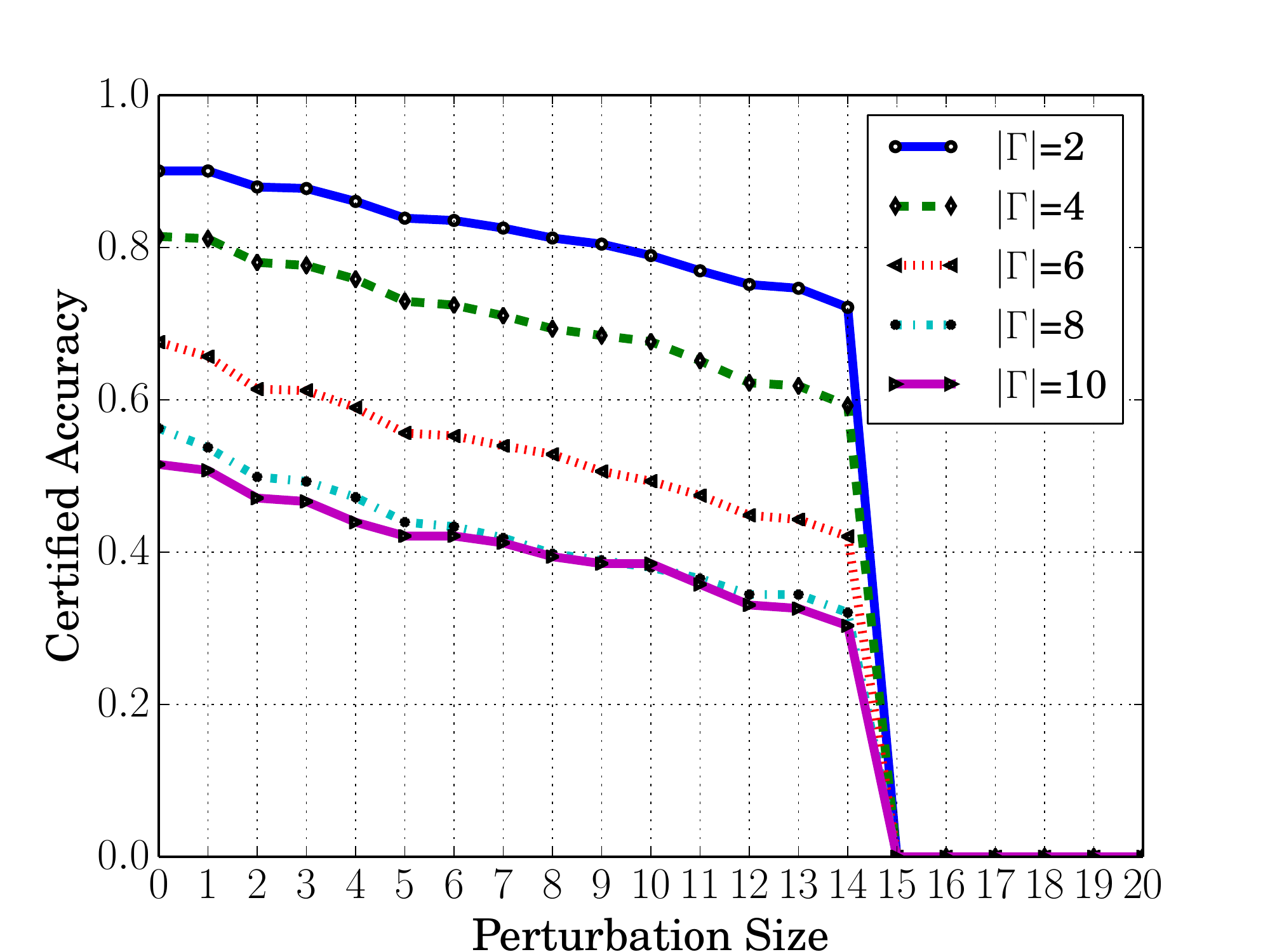} \label{fig2}} 
\subfloat[Amazon]{\includegraphics[width=0.32\textwidth]{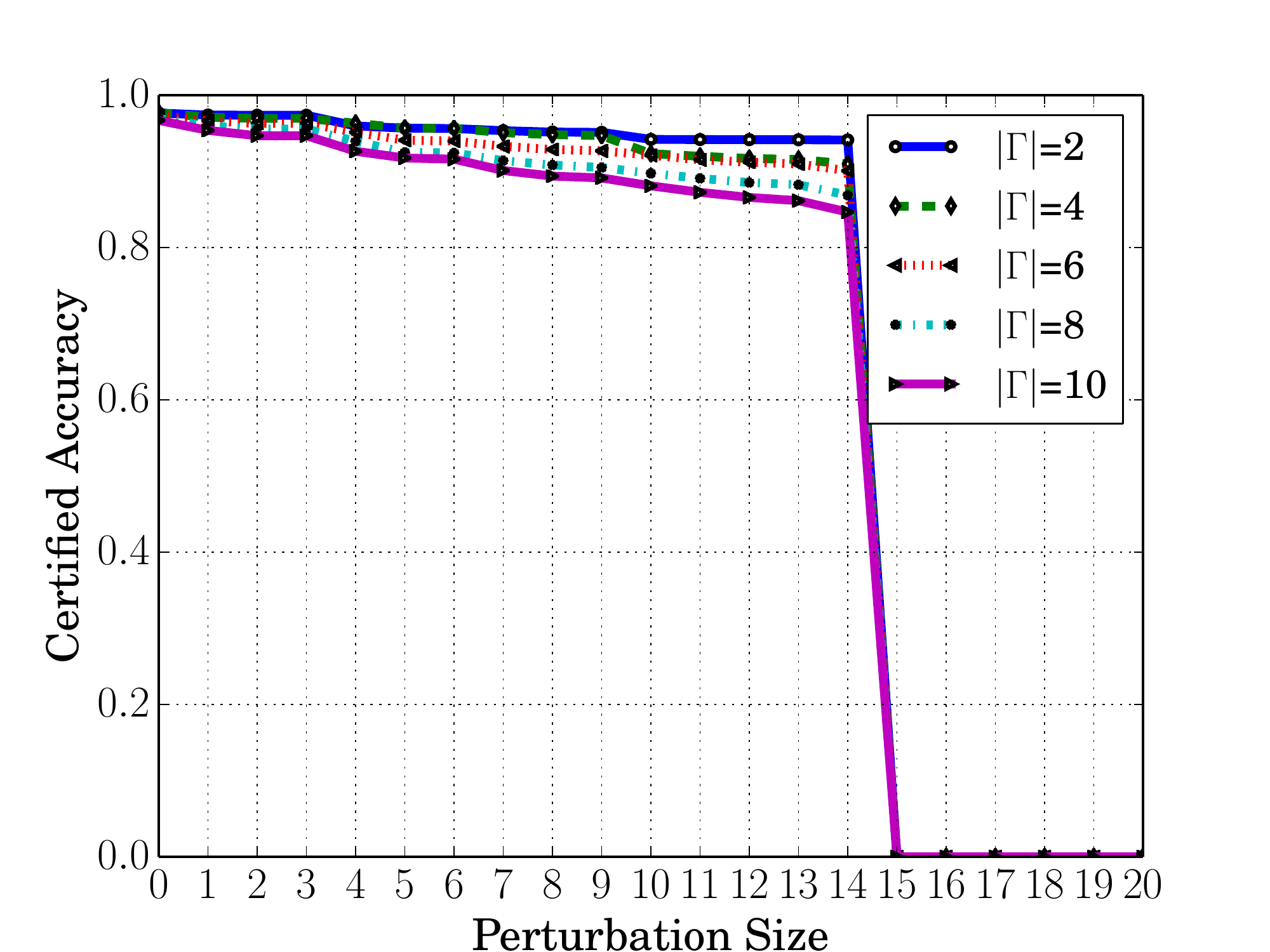} \label{fig2}} 
\caption{Impact of the number of victim nodes $|\Gamma|$ on defending against splitting attacks.}
\label{impact_setsize}
\end{figure*}

\begin{figure*}[!h]
\center
\subfloat[Email]{\includegraphics[width=0.32\textwidth]{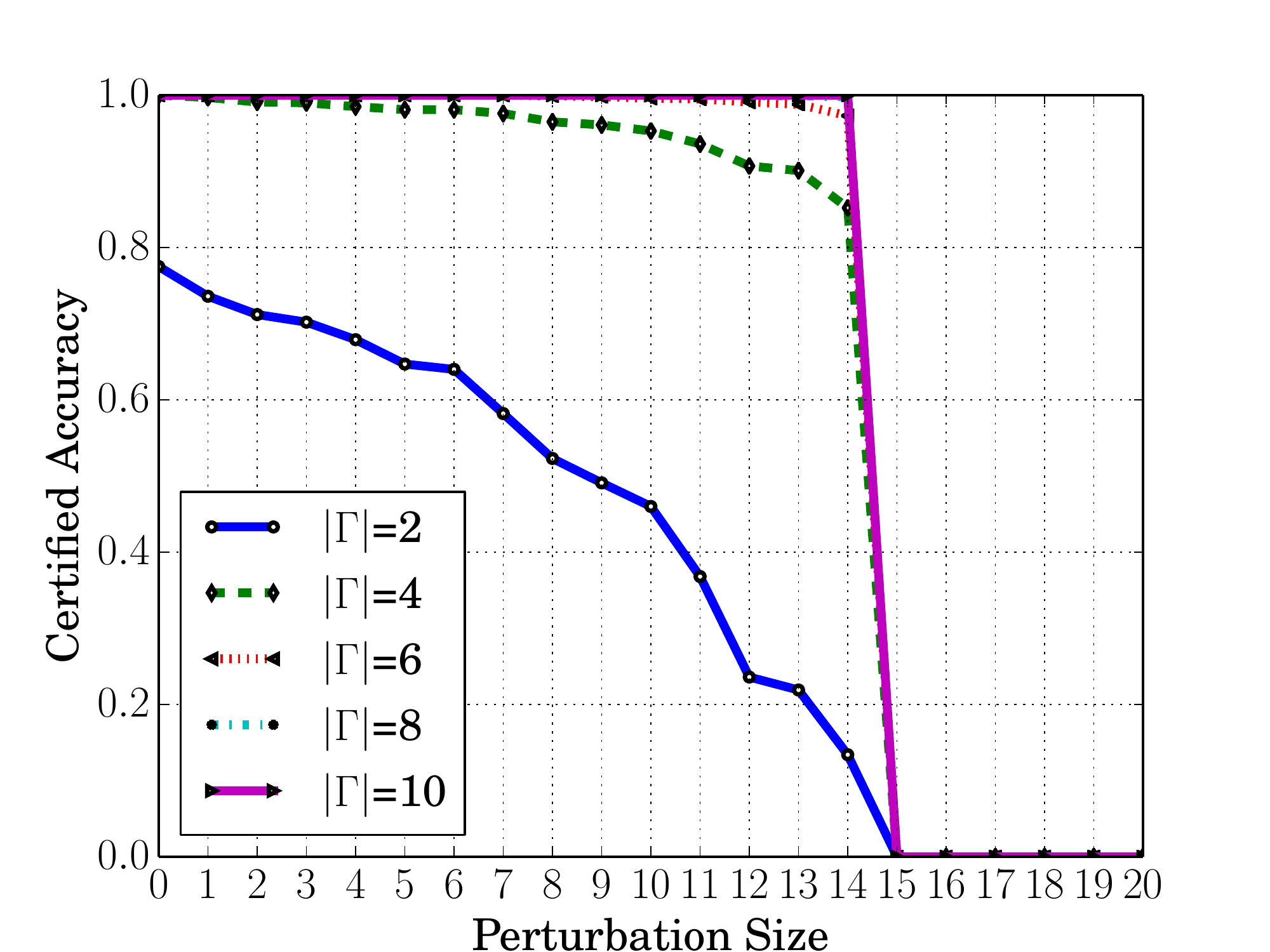} \label{fig1}} 
\subfloat[DBLP]{\includegraphics[width=0.32\textwidth]{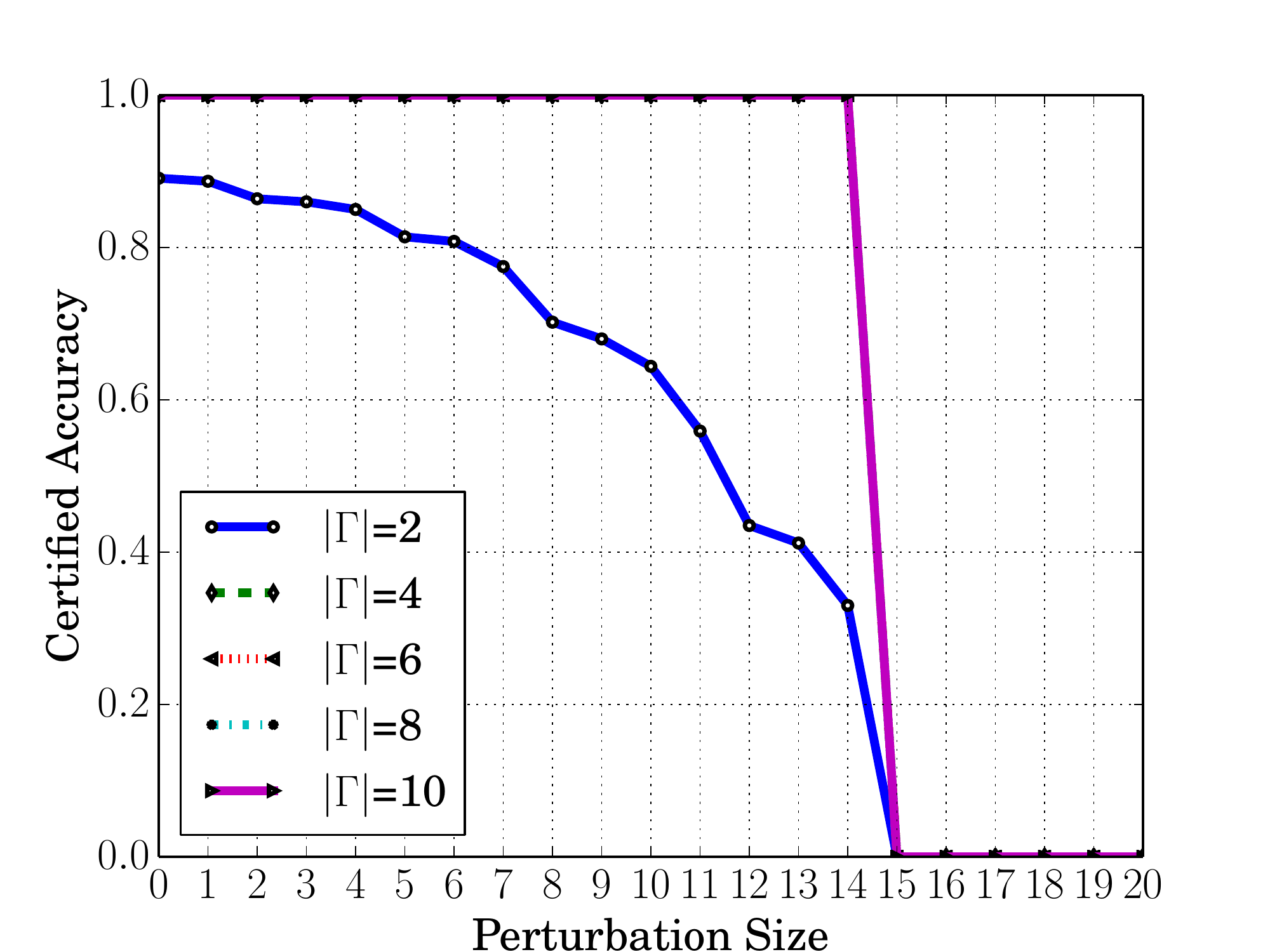} \label{fig2}} 
\subfloat[Amazon]{\includegraphics[width=0.32\textwidth]{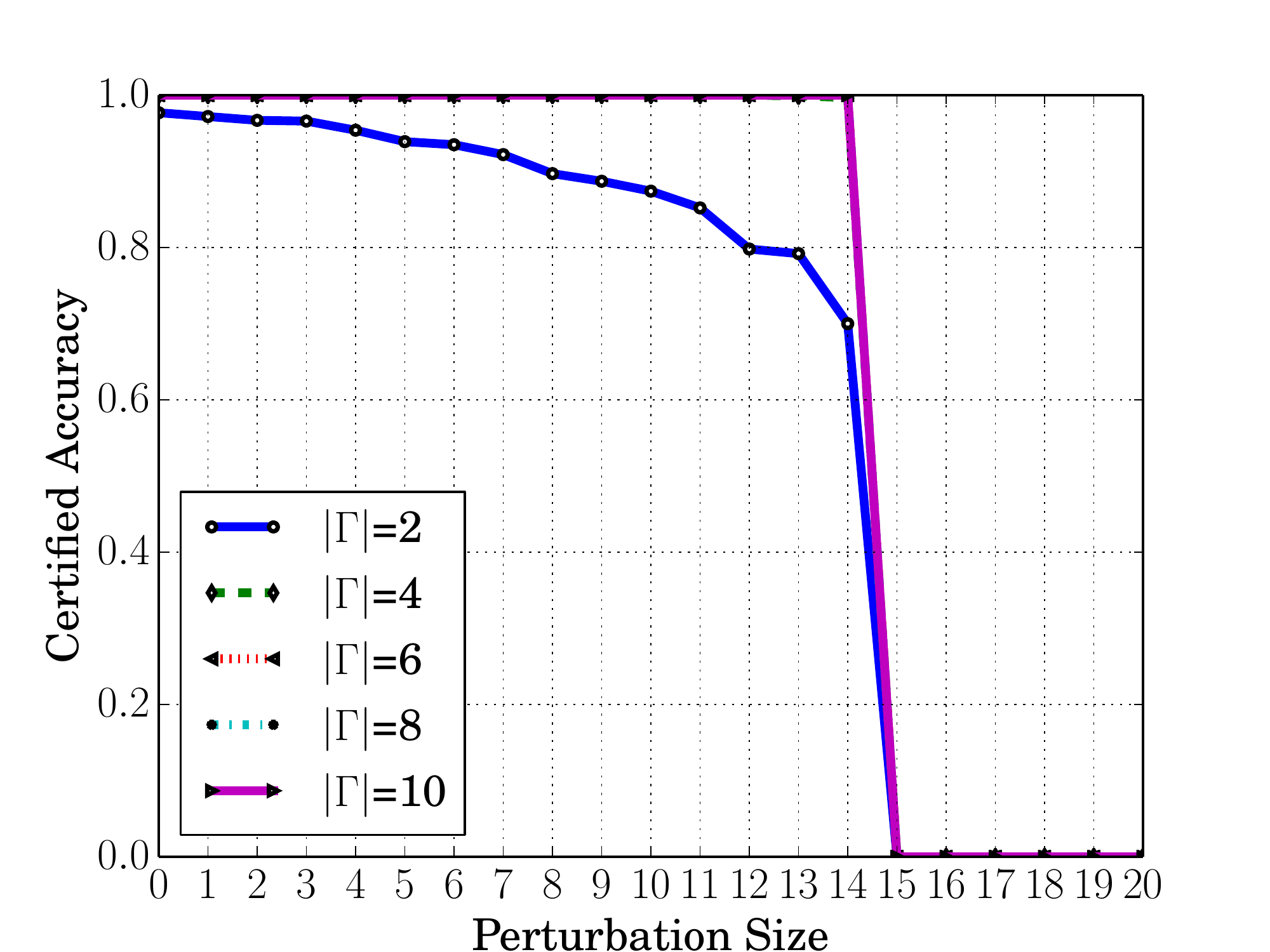} \label{fig2}} 
\caption{Impact of the number of victim nodes $|\Gamma|$ on defending against merging attacks.}
\label{impact_setsize1}
\end{figure*}

\subsection{Experimental Setup}

\myparatight{Datasets} We consider three undirected graph datasets with ``ground-truth'' communities, i.e., Email, DBLP, and Amazon. Table 1 shows the datasets and their statistics. We obtained the datasets from SNAP (http://snap.stanford.edu/). 

{\bf Email.} Email dataset describes the communications between members in a research institution. The graph consists of 1,005 nodes, each of which represents a member in the institution; and 25,571 edges, indicating the email communications between  members. The 42 departments in the institution are considered as the ground-truth communities and each node belongs to exactly one of them.

{\bf DBLP.} DBLP dataset contains 317,080 nodes and 1,049,866 edges. Each node represents a researcher in the computer science field. If two researchers co-authored at least one paper, then there is an edge between them in the graph. Every connected component with no less than 3 nodes is marked as a  community, which results in a total of 13,477 distinct communities. 

{\bf Amazon.} Amazon dataset is a network representing relations between products on Amazon website. The graph has 334,863 nodes and 925,872 edges. The nodes represent different products and two products are connected if they are frequently bought together. Overall, 75,149  communities are considered, each consisting of a connected component within the same category of products.

\myparatight{Community detection algorithm}
We use the popular Louvain's method~\cite{blondel2008fast} to detect communities. 
The method optimizes  modularity in a heuristic and iterative way. We note that the method produces  communities in multiple hierarchical levels, and we take the last level since in which the maximum of the modularity is attained. We use a publicly available implementation.\footnote{https://sites.google.com/site/findcommunities/}

\begin{figure*}[!h]
\center
\subfloat[Email]{\includegraphics[width=0.32\textwidth]{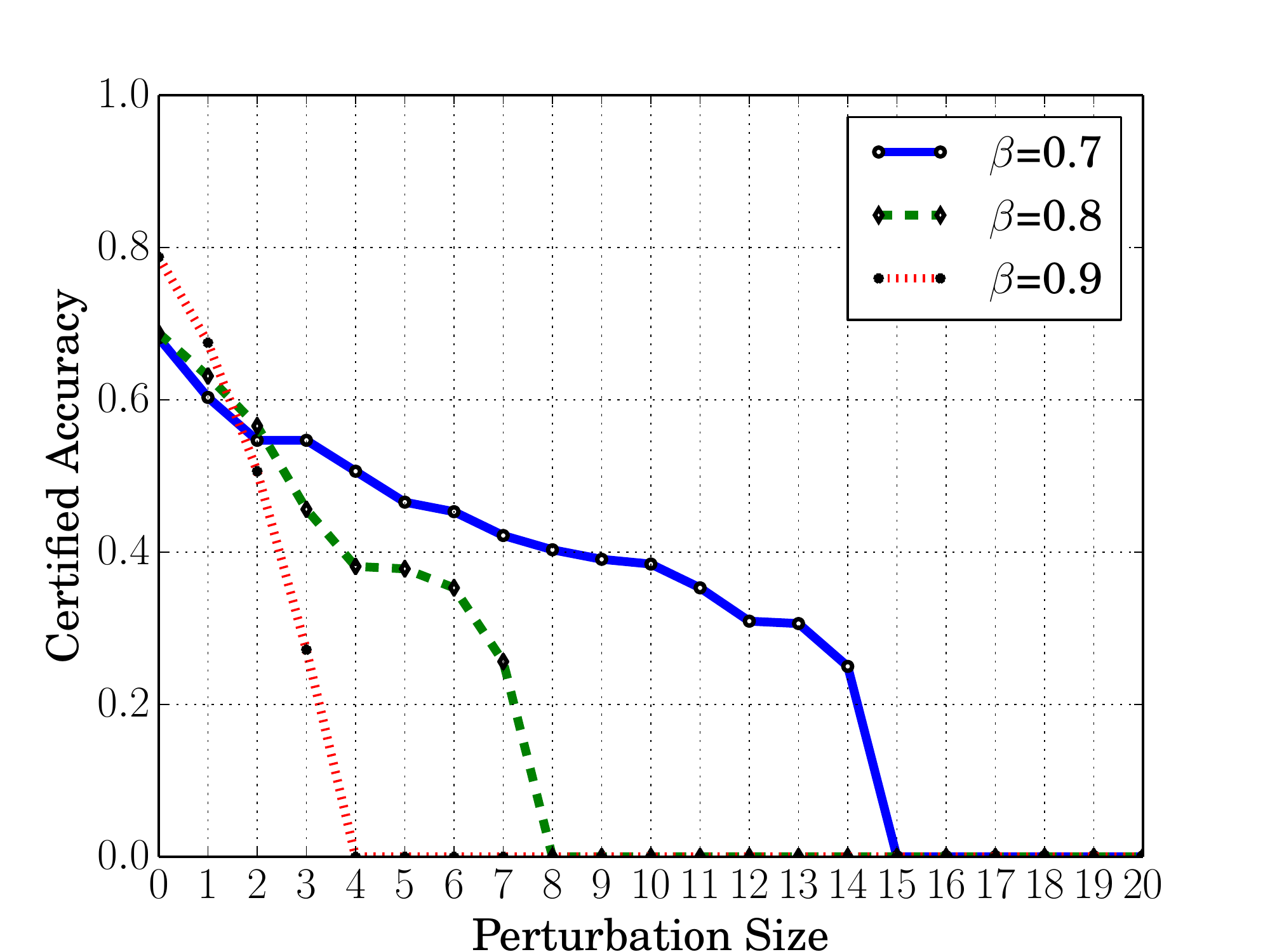} \label{fig1}} 
\subfloat[DBLP]{\includegraphics[width=0.32\textwidth]{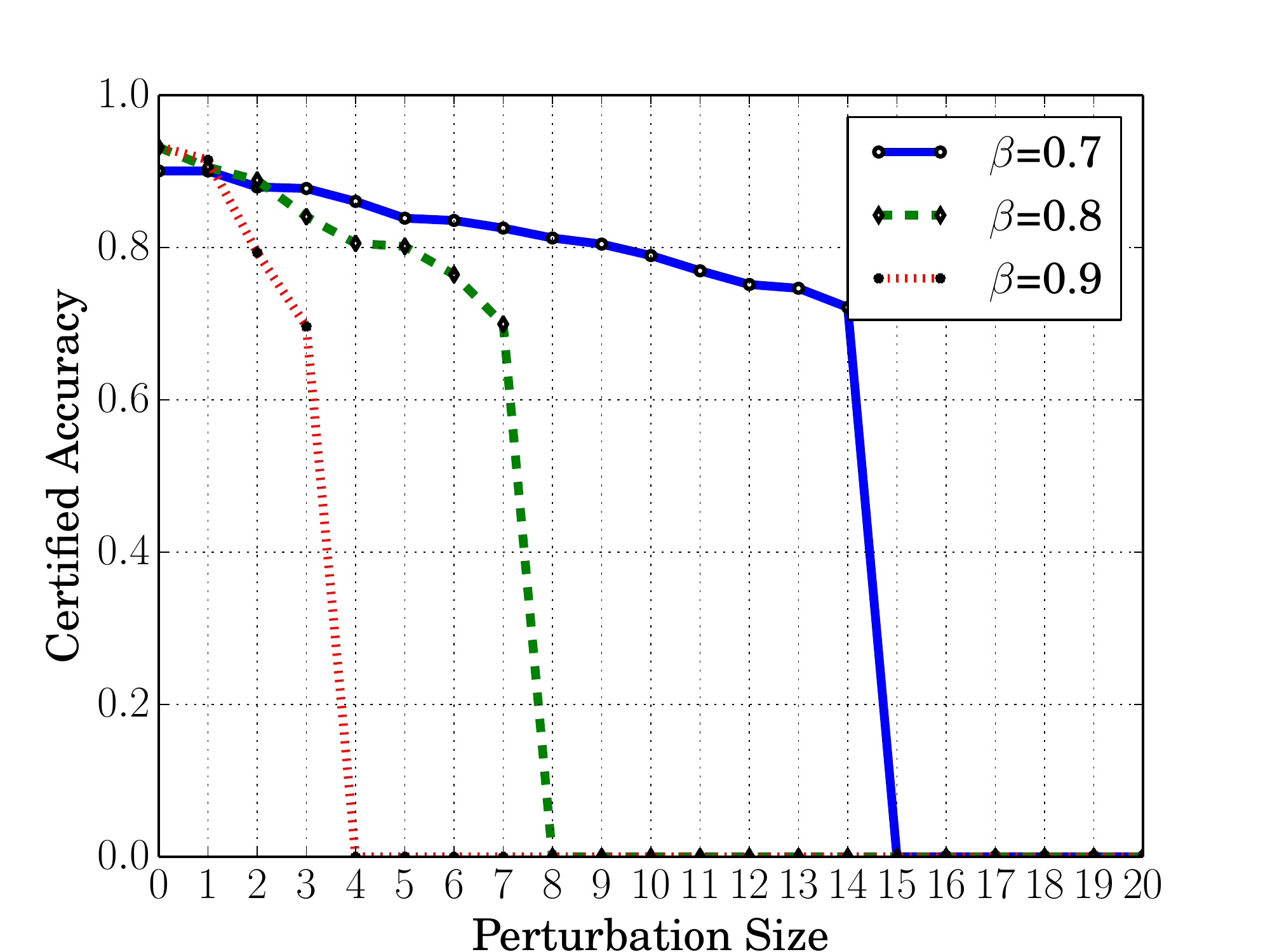} \label{fig2}} 
\subfloat[Amazon]{\includegraphics[width=0.32\textwidth]{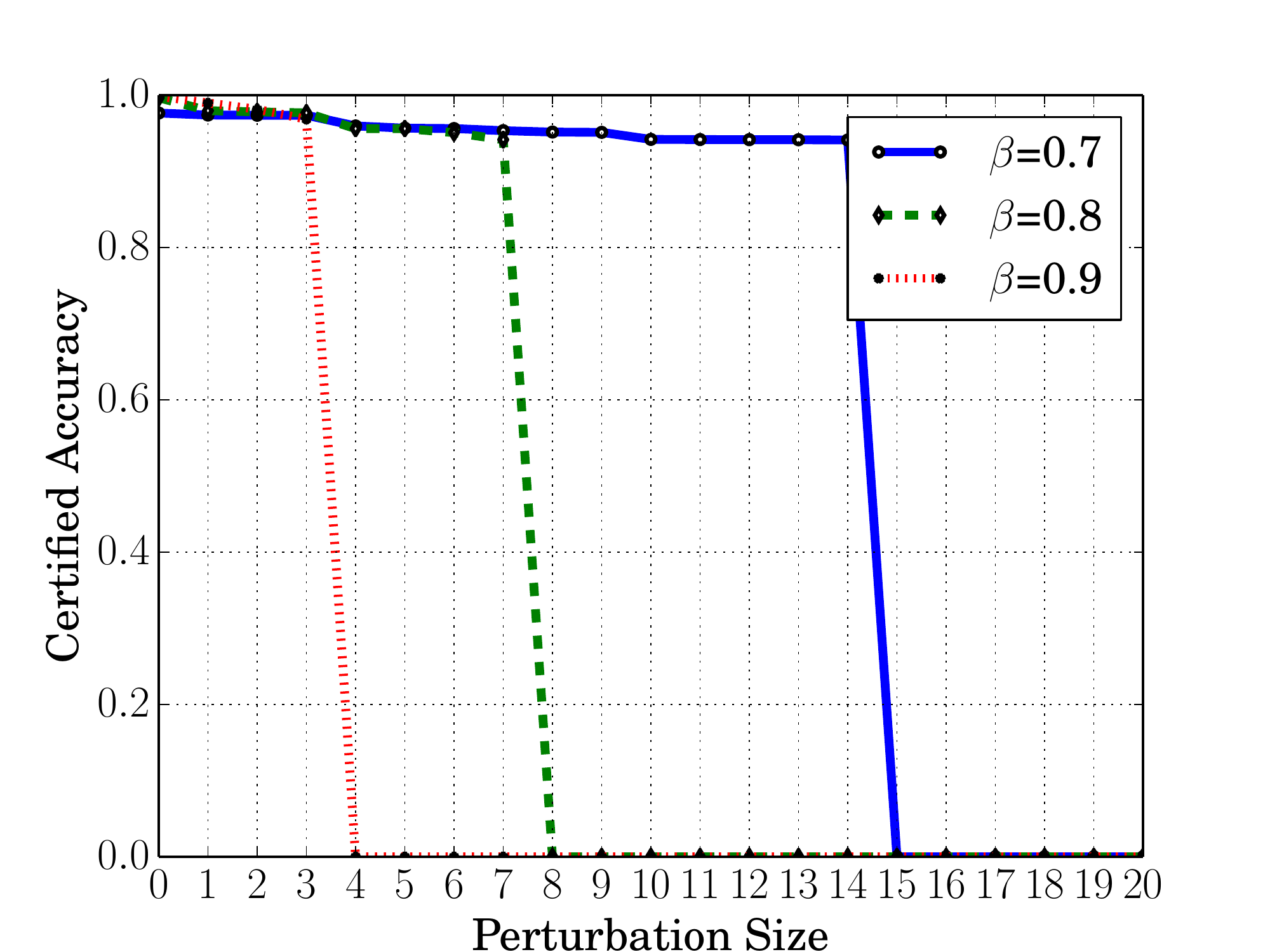} \label{fig2}} 
\caption{Impact of the noise parameter $\beta$ on defending against splitting attacks.}
\label{impact_noise_para}
\end{figure*}

\begin{figure*}[!h]
\center
\subfloat[Email]{\includegraphics[width=0.32\textwidth]{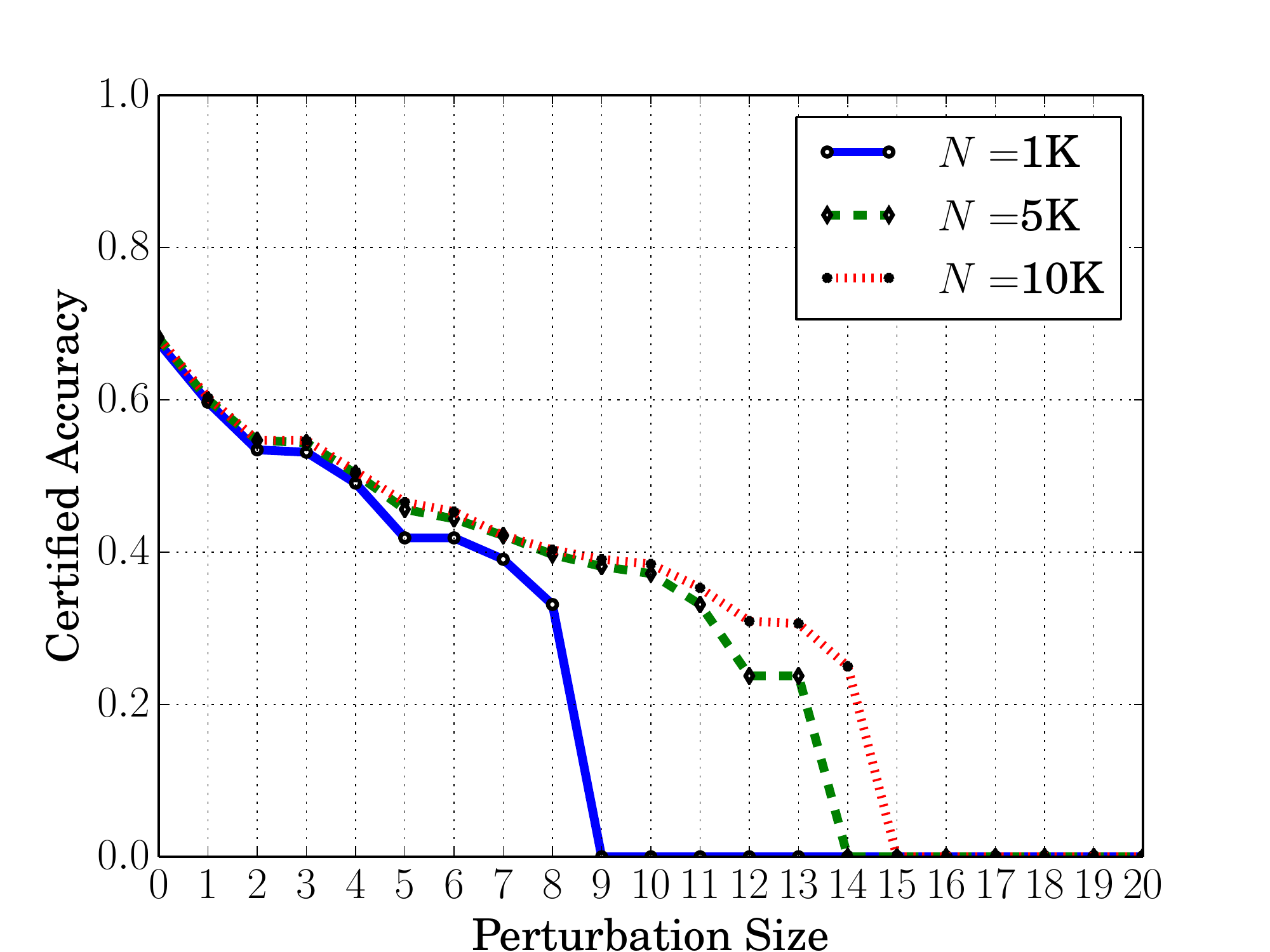} \label{fig1}} 
\subfloat[DBLP]{\includegraphics[width=0.32\textwidth]{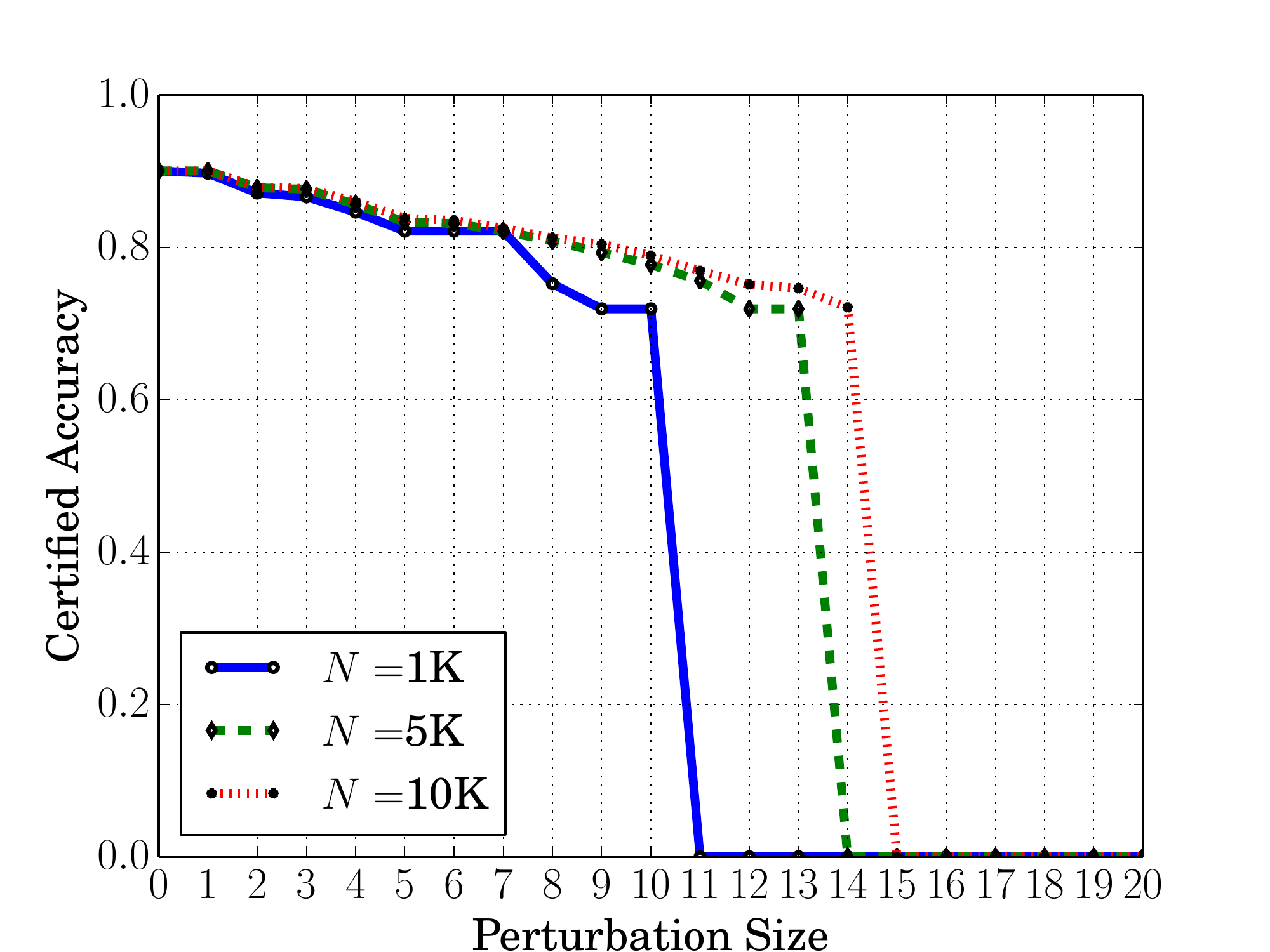} \label{fig2}} 
\subfloat[Amazon]{\includegraphics[width=0.32\textwidth]{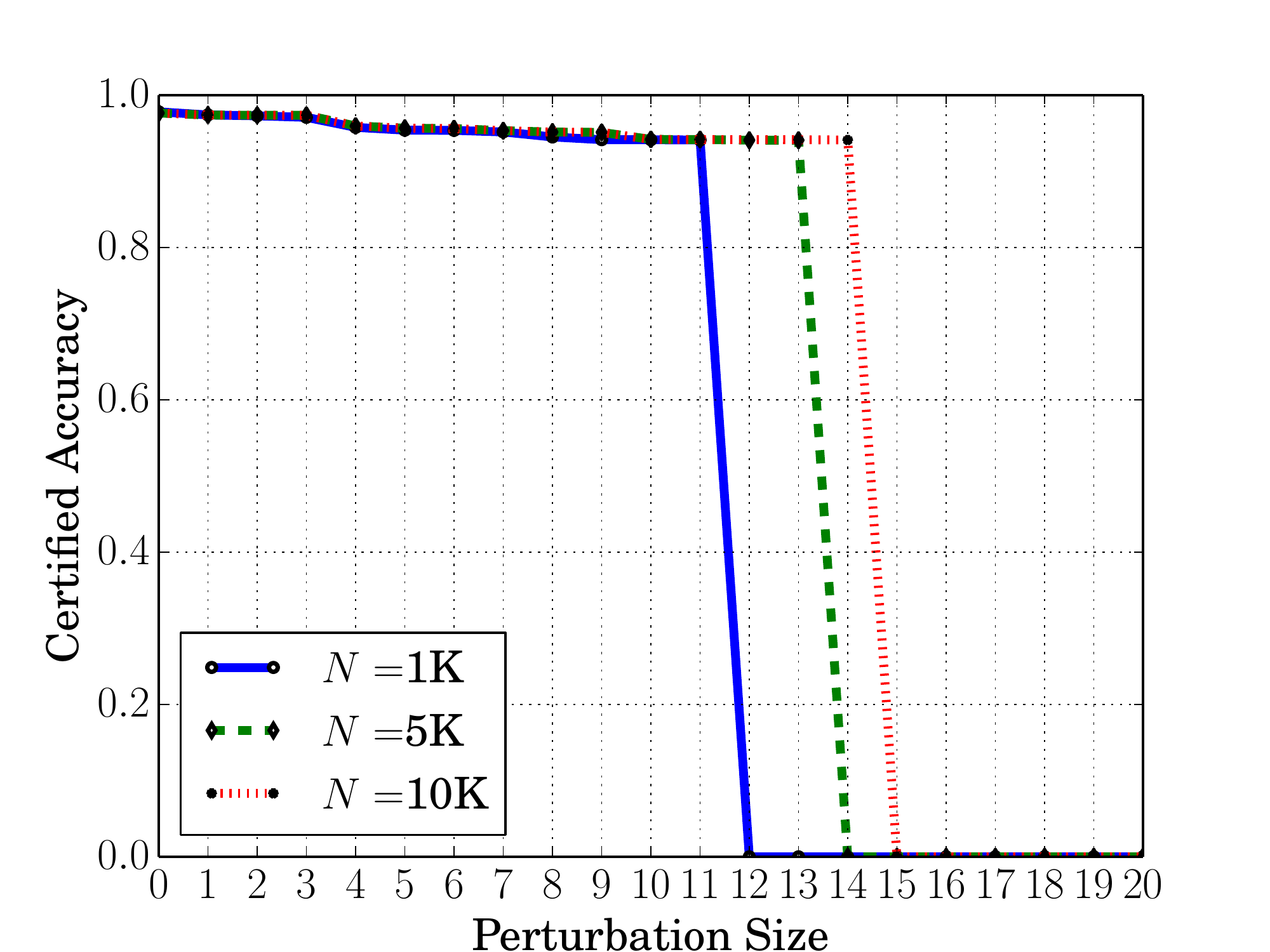} \label{fig2}} 
\caption{Impact of the number of noise samples $N$ on defending against splitting attacks.}
\label{impact_samples}
\end{figure*}

\myparatight{Evaluation metric}
We use \emph{certified accuracy} as the metric to evaluate our certification method. We take defending against the splitting attack as an example to illustrate certified accuracy. 
Suppose we are given $M$ sets of victim nodes $\Gamma_1,\Gamma_2,\cdots,\Gamma_M$. The nodes in each victim set $\Gamma_i$ are in the same ground-truth community. The goal of a splitting attack is to  perturb the graph structure such that the Louvain's method groups the victim nodes in a set $\Gamma_i$ into at least two communities. Our certification algorithm in Algorithm~\ref{alg:certify} produces an output $y_{i}$ and a certified perturbation size $L_i$ for each victim set $\Gamma_i$. $y_{i}=1$ means that we can provably guarantee that the nodes in  $\Gamma_i$ are grouped into the same community. Given a perturbation size $l$, we define the certified accuracy $CA(l)$ at the perturbation size $l$ as the fraction of sets of victim nodes whose output $y_{i}=1$ and certified perturbation size is at least $l$. Our certified accuracy $CA(l)$ is the fraction of sets of victim nodes that our method can provably detect as in the same community when an attacker adds or removes at most $l$ edges in the graph. Formally, we have:
\begin{align}
\text{\bf Certified } &\text{\bf Accuracy for Defending against Splitting Attacks:} \nonumber \\
&CA(l)=\frac{\sum_{i=1}^M \mathbb{I}(y_{i}=1) \mathbb{I}(L_i\geq l)}{M},
\end{align}
where $\mathbb{I}$ is an indicator function. 
For merging attacks, the nodes in a victim set $\Gamma_i$ are in different ground-truth communities. The goal of a merging attack is to  perturb the graph structure such that the Louvain's method groups the victim nodes in a set $\Gamma_i$ into the same community.
Given a perturbation size $l$, we define the certified accuracy $CA(l)$ at the perturbation size $l$ as the fraction of sets of victim nodes whose output $y_{i}=0$ and certified perturbation size is at least $l$. Our certified accuracy $CA(l)$ is the fraction of sets of victim nodes that our method can provably detect as in more than one communities when an attacker adds or removes at most $l$ edges in the graph. Formally, we have:
\begin{align}
\text{\bf Certified } &\text{\bf Accuracy for Defending against Merging Attacks:} \nonumber \\
&CA(l)=\frac{\sum_{i=1}^M \mathbb{I}(y_{i}=0) \mathbb{I}(L_i\geq l)}{M}.
\end{align}

\myparatight{Parameter setting}
Our method has the following parameters: the noise parameter $\beta$, the confidence level $1-\alpha$, and the number of samples $N$.
Unless otherwise mentioned, we use the following default parameters: $\beta=0.7$, $1-\alpha=0.999$, and $N=10,000$.  To estimate certified accuracy for defending against splitting attacks, we randomly sample two sets of $|\Gamma|$ nodes from each ground-truth community whose size is larger than $|\Gamma|$, and we treat them as victim sets. To estimate certified accuracy for defending against merging attacks, we randomly sample 1,000 victim sets, each of which includes nodes randomly sampled from $|\Gamma|$ different communities.   By default, we assume each set of victim nodes includes 2 nodes, i.e., $|\Gamma|=2$. We also study the impact of each parameter, including $\beta$, $1-\alpha$, $N$, and $|\Gamma|$.   
When studying the impact of one parameter, we fix the remaining parameters to be their default values. 
We randomly pick $100$ nodes as attacker-controlled nodes for each dataset, and the attacker perturbs the edges between them.

\begin{figure*}[!h]
\center
\subfloat[Email]{\includegraphics[width=0.32\textwidth]{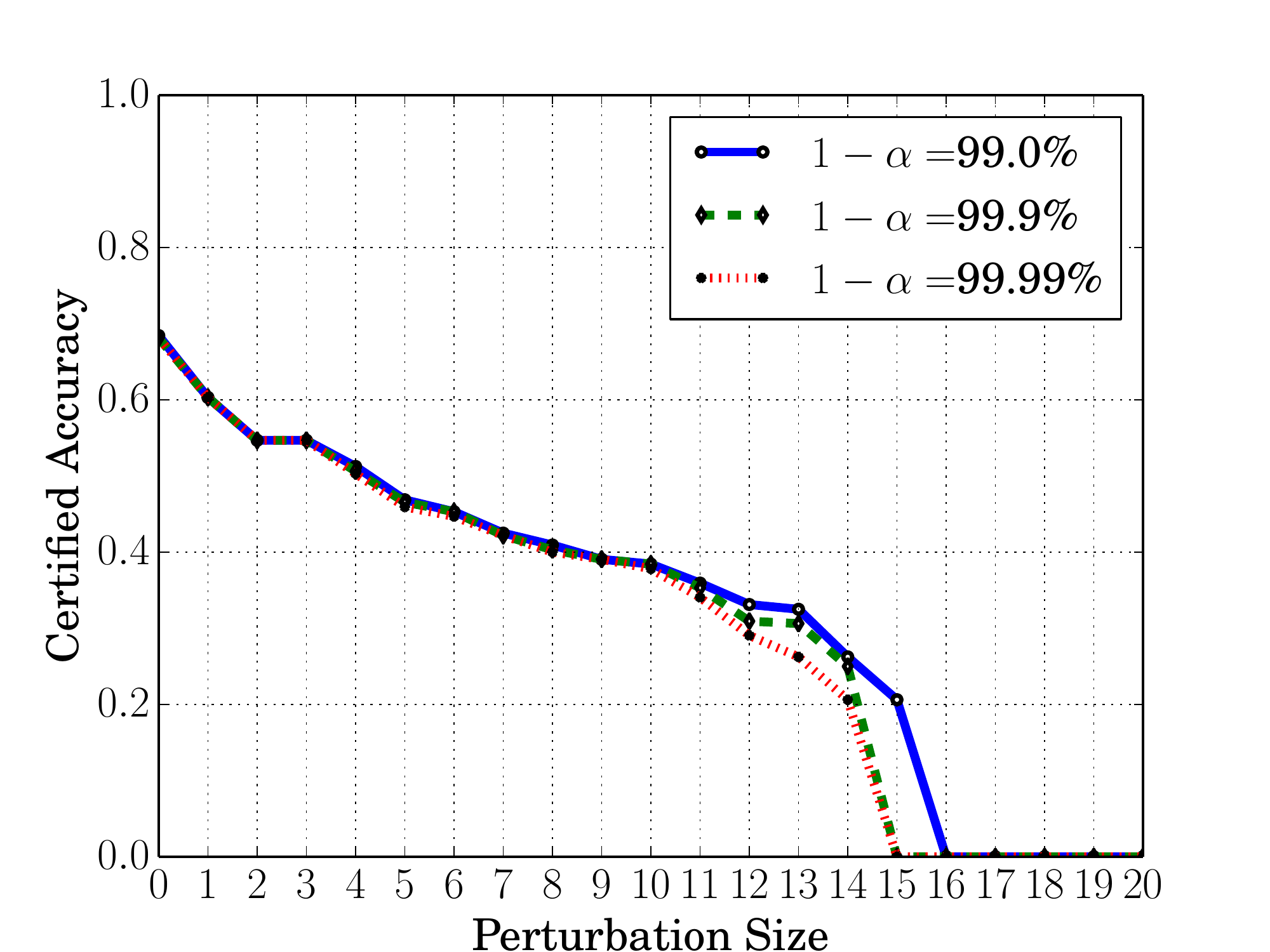} \label{fig1}} 
\subfloat[DBLP]{\includegraphics[width=0.32\textwidth]{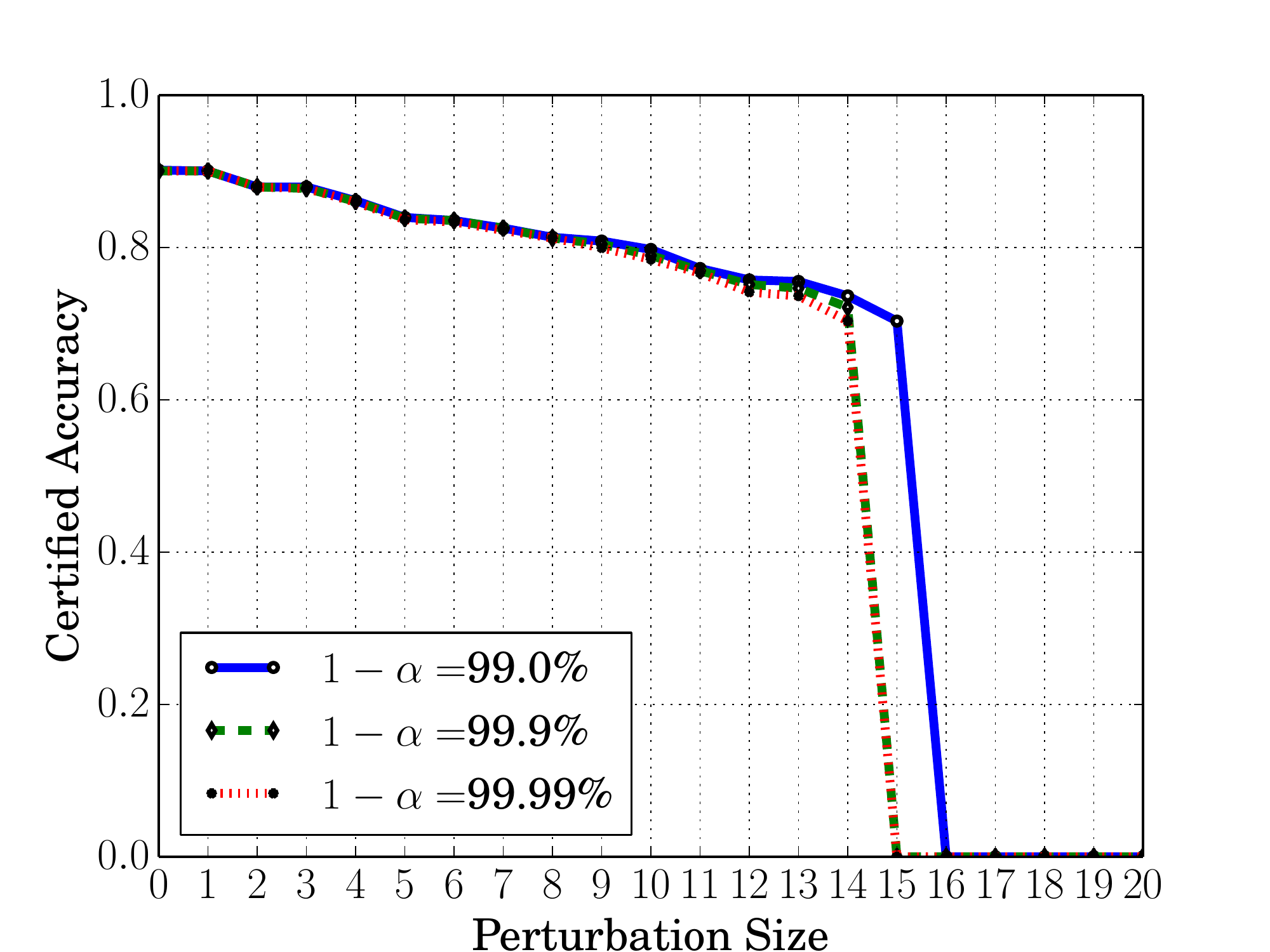} \label{fig2}} 
\subfloat[Amazon]{\includegraphics[width=0.32\textwidth]{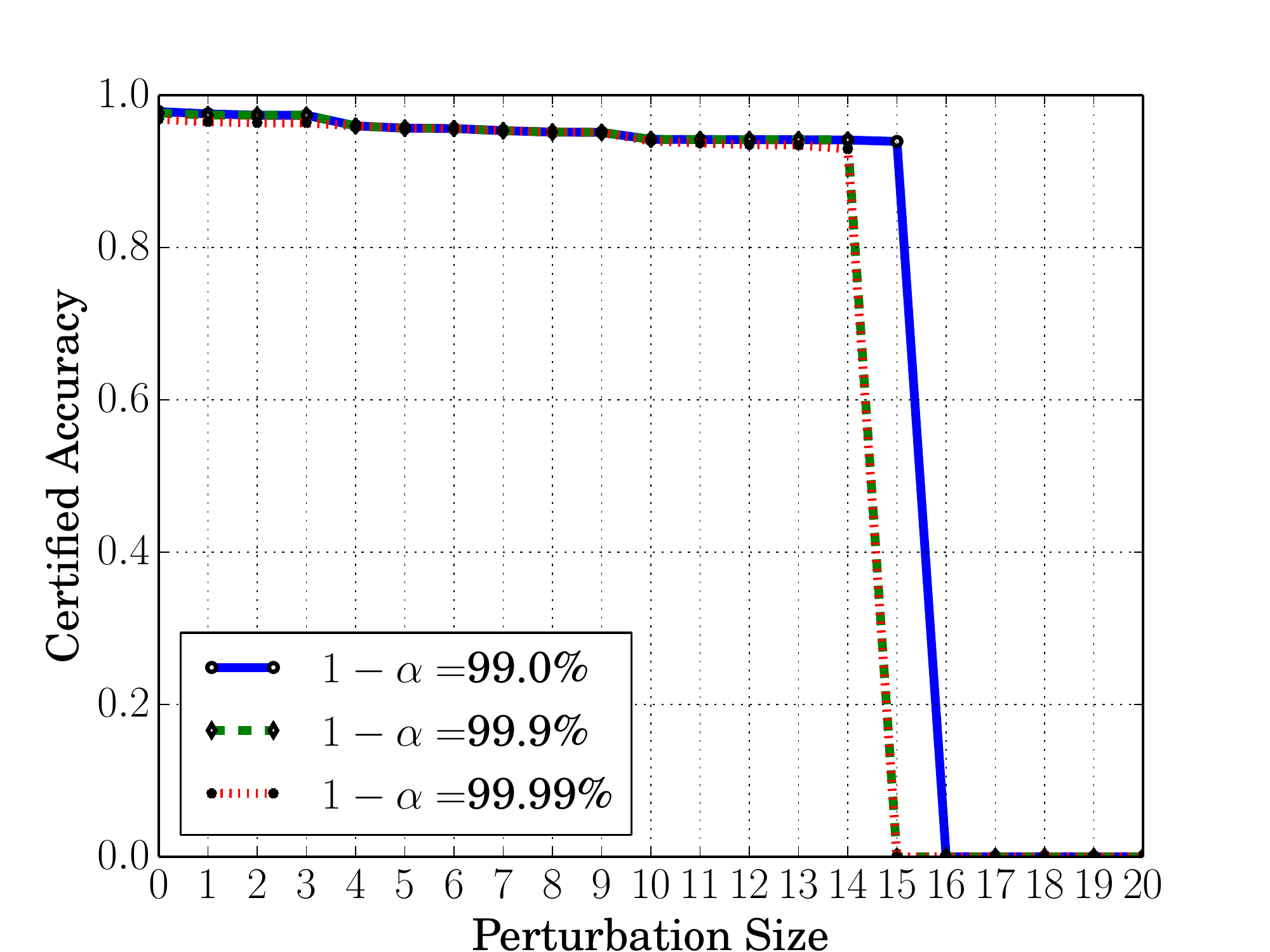} \label{fig2}} 
\caption{Impact of the confidence level $1-\alpha$ on defending against splitting attacks.}
\label{impact_alpha}
\end{figure*}

\subsection{Experimental Results}

\myparatight{Impact of the number of victim nodes $|\Gamma|$}
Figure~\ref{impact_setsize} shows the certified accuracy vs. perturbation size for defending against splitting attacks with different number of victim nodes on the three datasets, while Figure~\ref{impact_setsize1} shows the results for defending against merging attacks. We observe that as the number of victim nodes increases, the curve of the certified accuracy becomes lower for splitting attacks and higher for merging attacks. This is because it is harder to provably guarantee that a larger set of nodes are detected as in the same community (defending against splitting attacks); and it is easier to provably guarantee that a larger set of nodes are detected as in more than one communities (defending against merging attacks).

\myparatight{Impact of the noise parameter $\beta$}
Figure~\ref{impact_noise_para} shows the certified accuracy vs. perturbation size for defending against splitting attacks with different noise parameter $\beta$ on the three datasets. We observe that $\beta$ provides a tradeoff between normal accuracy without attacks and robustness. Specifically, when $\beta$ is larger, the normal accuracy, i.e., certified accuracy at perturbation size  0, is larger, while the certified  accuracy decreases more quickly as the perturbation size increases.  We also have similar observations of the certified accuracy vs. perturbation size for defending against merging attacks, and thus we omit the results for simplicity.

\myparatight{Impact of the number of sampled noise $N$}
Figure~\ref{impact_samples} shows the certified  accuracy vs. perturbation size for defending against splitting attacks with different numbers of sampled noise $N$ on the three datasets. We observe that the curve is higher as  $N$ increases. This is because a larger  $N$  makes the estimated probability bound $\underline{p}$ tighter and thus the certified perturbation size is also larger. We also have similar observations of the certified accuracy vs. perturbation size for defending against merging attacks, and thus we omit the results for simplicity.

\myparatight{Impact of the confidence level $1 - \alpha$}
Figure~\ref{impact_alpha} shows the certified  accuracy  vs. perturbation size for defending against splitting attacks with different confidence levels $1-\alpha$ on the three datasets. We observe that as the confidence level $1-\alpha$ increases, the curve of the certified  accuracy becomes lower. The reason is that a higher confidence level causes a looser estimated probability bound  $\underline{p}$ and thus the certified perturbation size is smaller. 
However, we note that the differences of the certified  accuracies between different confidence levels are negligible when the confidence levels are large enough. We also have similar observations of the certified accuracy vs. perturbation size for defending against merging attacks, and thus we omit the results for simplicity.

\section{Related Work}

\myparatight{Adversarial attacks to non-graph data and their defenses} 
For non-graph data,  \emph{adversarial example} is a well-known adversarial attack. Specifically, an attacker adds a carefully crafted perturbation to an input example such that a machine learning classifier makes predictions for the perturbed example as the attacker desires. The input example with carefully crafted perturbation is called adversarial example~\cite{Szegedy14,goodfellow2014explaining}. 
Various empirical defenses (e.g.,~\cite{papernot2016distillation,madry2017towards,goodfellow2014explaining}) have been proposed to defend against  adversarial examples. However,  these defenses were often soon broken by adaptive attacks~\cite{carlini2017adversarial,athalye2018obfuscated}. 

In response, various certified defenses (e.g.,~\cite{scheibler2015towards,cheng2017maximum,raghunathan2018certified,gehr2018ai2,weng2018towards}) against adversarial examples have been developed. 
Among these methods, randomized smoothing is state-of-the-art. Randomized smoothing turns an arbitrary classifier/function into a robust one via adding random noise to the input. Randomized smoothing was initially proposed as empirical defenses~\cite{cao2017mitigating,liu2018towards} without formal certified robustness guarantees. For example, Cao \& Gong~\cite{cao2017mitigating} proposed to use uniform noise sampled from a hypercube centered at an input. Lecuyer et al.~\cite{lecuyer2018certified} derived the first certified robustness guarantee for randomized smoothing with Gaussian or Laplacian noise by utilizing  differential privacy techniques. Li et al.~\cite{li2018second} further improved the certified robustness bound by using techniques in information theory. In particular, they applied data processing inequality to R\'enyi divergence and obtained certified robustness guarantee for Gaussian noise. Cohen et al.~\cite{cohen2019certified} proved the first tight bound for  certified robustness under isotropic Gaussian noise. 
Jia et al.~\cite{jia2020certified} derived a tight certified robustness of top-k predictions for randomized smoothing with Gaussian noise. In particular, they showed that a label is among the top-k labels predicted by the smoothed classifier when the adversarial perturbation is bounded.  
Our work uses randomized smoothing. However, different from the existing randomized smoothing methods, which assume continuous input and add continuous noise, we propose randomized smoothing on binary data and leverage it to certify robustness of community detection against splitting and merging attacks. 
 We note that a concurrent work~\cite{lee2019tight} generalized randomized smoothing to discrete data. The major difference between our approach and~\cite{lee2019tight} is that we leverage a variant of the Neyman-Pearson Lemma to derive the certified perturbation size. 

\myparatight{Adversarial attacks to graph data and their defenses} Compared to non-graph data, adversarial attacks to graph data and their defenses are much less studied.  \emph{Adversarial structural perturbation} is a recently proposed attack to graph data. For instance, several recent studies~\cite{dai2018adversarial,zugner2018adversarial,zugner2019adversarial,bojchevski2019adversarial,wang2019attacking} showed that Graph Neural Networks (GNNs) are vulnerable to adversarial structural perturbations. Specifically, an attacker can slightly perturb the graph structure and/or node features to mislead the predictions made by GNNs. Some empirical defenses~\cite{wu2019adversarial,xu2019topology,zhu2019robust} were proposed to defend against such attacks. However, these methods do not have certified robustness guarantees. Z{\"u}gner \& G{\"u}nnemann~\cite{Zugner2019Certifiable} developed the first certified robustness guarantee against node-feature perturbations for graph convolutional network~\cite{kipf2017semi}. Bojchevski \& G{\"u}nnemann~\cite{bojchevski2019certifiable} proposed the first method for verifying certifiable (non-)robustness of graph convolutional network against structural  perturbations. These work is different from ours as we focus on certifying robustness of community detection.

Multiple studies~\cite{nagaraja2010impact,waniek2018hiding,fionda2017community,chen2019ga,chen2017practical} have shown that community detection is vulnerable to adversarial structural perturbation. 
Several heuristic defenses~\cite{nagaraja2010impact,chen2017practical} were proposed to enhance the robustness of community detection against adversarial structural perturbations. However, these defenses lack formal guarantees. Our work is the first certified robustness guarantee of community detection against adversarial structural perturbations.

\section{Discussion and Limitations}
Given a set of nodes that are in the same ground-truth community (or in different ground-truth communities),  
our certified robustness guarantees that the nodes are provably detected as in the same community (or in different communities) when the number of added or removed edges in the graph is at most a certain threshold (called certified perturbation size). We note that when we add or remove enough edges in a graph, the ``ground-truth'' communities may change, and thus we may expect the set of nodes to be detected as in different communities (or in the same community). Therefore, the certified perturbation size should not be too large. We believe it is an interesting future work to explore what certified perturbation size should be expected for a particular application scenario. 

Our work shows that we can provably guarantee that a set of nodes are or are not in the same community when an attacker adds or deletes a bounded number of edges in the graph. However, it is still an open question on how to obtain communities from our smoothed community detection method.   One possible way to obtain communities is as follows: we first randomly pick a node as the initial community C. For each remaining node, we compute the probability of the node being clustered into the same community with each node in C under randomized smoothing. Then, we compute the average probability and if it is larger than a threshold, we add the node to C. When no more nodes can be added to C, we randomly pick another node from the remaining nodes and repeat the above process until all nodes are clustered into certain communities. We believe it is an interesting future work to explore how to derive communities from the smoothed method. We note that the communities derived from the smoothed community detection method may be less accurate than those derived from the base community detection method. In other words, there may be a tradeoff between accuracy and robustness.

\section{Conclusion}
In this work, we develop the first certified robustness guarantee of  community detection against adversarial structural perturbations. Specifically, our results show that a set of nodes can be  provably detected as in the same community (against splitting attacks) or in different communities (against merging attacks) when the number of edges added or removed by an attacker is no larger than a threshold. Moreover, we show that our derived threshold is tight when randomized smoothing with our discrete noise is used.  Our method can turn any community detection method to be provably robust against adversarial structural perturbation to defend against splitting and merging attacks. We also empirically demonstrate the effectiveness of our method using three real-world graph datasets with ground-truth communities.  
 Interesting future work includes  leveraging the information of the community detection algorithm to further improve the certified robustness guarantees and exploring what certified perturbation size should be expected for a particular application scenario.

\noindent\textbf{ACKNOWLEDGMENTS}\\
We thank the anonymous reviewers for insightful reviews. This work was supported by NSF grant No. 1937787 and No. 1937786. 

\bibliographystyle{ACM-Reference-Format}
\bibliography{refs}


\begin{thebibliography}{44}


\ifx \showCODEN    \undefined \def \showCODEN     #1{\unskip}     \fi
\ifx \showDOI      \undefined \def \showDOI       #1{#1}\fi
\ifx \showISBNx    \undefined \def \showISBNx     #1{\unskip}     \fi
\ifx \showISBNxiii \undefined \def \showISBNxiii  #1{\unskip}     \fi
\ifx \showISSN     \undefined \def \showISSN      #1{\unskip}     \fi
\ifx \showLCCN     \undefined \def \showLCCN      #1{\unskip}     \fi
\ifx \shownote     \undefined \def \shownote      #1{#1}          \fi
\ifx \showarticletitle \undefined \def \showarticletitle #1{#1}   \fi
\ifx \showURL      \undefined \def \showURL       {\relax}        \fi
\providecommand\bibfield[2]{#2}
\providecommand\bibinfo[2]{#2}
\providecommand\natexlab[1]{#1}
\providecommand\showeprint[2][]{arXiv:#2}

\bibitem[\protect\citeauthoryear{Athalye, Carlini, and Wagner}{Athalye
  et~al\mbox{.}}{2018}]%
        {athalye2018obfuscated}
\bibfield{author}{\bibinfo{person}{Anish Athalye}, \bibinfo{person}{Nicholas
  Carlini}, {and} \bibinfo{person}{David Wagner}.}
  \bibinfo{year}{2018}\natexlab{}.
\newblock \showarticletitle{Obfuscated Gradients Give a False Sense of
  Security: Circumventing Defenses to Adversarial Examples}. In
  \bibinfo{booktitle}{\emph{International Conference on Machine Learning}}.
  \bibinfo{pages}{274--283}.
\newblock


\bibitem[\protect\citeauthoryear{Blondel, Guillaume, Lambiotte, and
  Lefebvre}{Blondel et~al\mbox{.}}{2008}]%
        {blondel2008fast}
\bibfield{author}{\bibinfo{person}{Vincent~D Blondel},
  \bibinfo{person}{Jean-Loup Guillaume}, \bibinfo{person}{Renaud Lambiotte},
  {and} \bibinfo{person}{Etienne Lefebvre}.} \bibinfo{year}{2008}\natexlab{}.
\newblock \showarticletitle{Fast unfolding of communities in large networks}.
\newblock \bibinfo{journal}{\emph{Journal of statistical mechanics: theory and
  experiment}} \bibinfo{volume}{2008}, \bibinfo{number}{10}
  (\bibinfo{year}{2008}), \bibinfo{pages}{P10008}.
\newblock


\bibitem[\protect\citeauthoryear{Bojchevski and G{\"u}nnemann}{Bojchevski and
  G{\"u}nnemann}{2019a}]%
        {bojchevski2019adversarial}
\bibfield{author}{\bibinfo{person}{Aleksandar Bojchevski} {and}
  \bibinfo{person}{Stephan G{\"u}nnemann}.} \bibinfo{year}{2019}\natexlab{a}.
\newblock \showarticletitle{Adversarial Attacks on Node Embeddings via Graph
  Poisoning}. In \bibinfo{booktitle}{\emph{ICML}}.
\newblock


\bibitem[\protect\citeauthoryear{Bojchevski and G{\"u}nnemann}{Bojchevski and
  G{\"u}nnemann}{2019b}]%
        {bojchevski2019certifiable}
\bibfield{author}{\bibinfo{person}{Aleksandar Bojchevski} {and}
  \bibinfo{person}{Stephan G{\"u}nnemann}.} \bibinfo{year}{2019}\natexlab{b}.
\newblock \showarticletitle{Certifiable Robustness to Graph Perturbations}. In
  \bibinfo{booktitle}{\emph{Advances in Neural Information Processing
  Systems}}. \bibinfo{pages}{8317--8328}.
\newblock


\bibitem[\protect\citeauthoryear{Brown, Cai, and DasGupta}{Brown
  et~al\mbox{.}}{2001}]%
        {brown2001interval}
\bibfield{author}{\bibinfo{person}{Lawrence~D Brown}, \bibinfo{person}{T~Tony
  Cai}, {and} \bibinfo{person}{Anirban DasGupta}.}
  \bibinfo{year}{2001}\natexlab{}.
\newblock \showarticletitle{Interval estimation for a binomial proportion}.
\newblock \bibinfo{journal}{\emph{Statistical science}} (\bibinfo{year}{2001}),
  \bibinfo{pages}{101--117}.
\newblock


\bibitem[\protect\citeauthoryear{Cao and Gong}{Cao and Gong}{2017}]%
        {cao2017mitigating}
\bibfield{author}{\bibinfo{person}{Xiaoyu Cao} {and}
  \bibinfo{person}{Neil~Zhenqiang Gong}.} \bibinfo{year}{2017}\natexlab{}.
\newblock \showarticletitle{Mitigating evasion attacks to deep neural networks
  via region-based classification}. In \bibinfo{booktitle}{\emph{Proceedings of
  the 33rd Annual Computer Security Applications Conference}}. ACM,
  \bibinfo{pages}{278--287}.
\newblock


\bibitem[\protect\citeauthoryear{Carlini and Wagner}{Carlini and
  Wagner}{2017}]%
        {carlini2017adversarial}
\bibfield{author}{\bibinfo{person}{Nicholas Carlini} {and}
  \bibinfo{person}{David Wagner}.} \bibinfo{year}{2017}\natexlab{}.
\newblock \showarticletitle{Adversarial examples are not easily detected:
  Bypassing ten detection methods}. In \bibinfo{booktitle}{\emph{Proceedings of
  the 10th ACM Workshop on Artificial Intelligence and Security}}. ACM,
  \bibinfo{pages}{3--14}.
\newblock


\bibitem[\protect\citeauthoryear{Chen, Chen, Chen, Zhao, Yu, Xuan, and
  Yang}{Chen et~al\mbox{.}}{2019}]%
        {chen2019ga}
\bibfield{author}{\bibinfo{person}{Jinyin Chen}, \bibinfo{person}{Lihong Chen},
  \bibinfo{person}{Yixian Chen}, \bibinfo{person}{Minghao Zhao},
  \bibinfo{person}{Shanqing Yu}, \bibinfo{person}{Qi Xuan}, {and}
  \bibinfo{person}{Xiaoniu Yang}.} \bibinfo{year}{2019}\natexlab{}.
\newblock \showarticletitle{GA-Based Q-Attack on Community Detection}.
\newblock \bibinfo{journal}{\emph{IEEE Transactions on Computational Social
  Systems}} \bibinfo{volume}{6}, \bibinfo{number}{3} (\bibinfo{year}{2019}),
  \bibinfo{pages}{491--503}.
\newblock


\bibitem[\protect\citeauthoryear{Chen, Nadji, Kountouras, Monrose, Perdisci,
  Antonakakis, and Vasiloglou}{Chen et~al\mbox{.}}{2017}]%
        {chen2017practical}
\bibfield{author}{\bibinfo{person}{Yizheng Chen}, \bibinfo{person}{Yacin
  Nadji}, \bibinfo{person}{Athanasios Kountouras}, \bibinfo{person}{Fabian
  Monrose}, \bibinfo{person}{Roberto Perdisci}, \bibinfo{person}{Manos
  Antonakakis}, {and} \bibinfo{person}{Nikolaos Vasiloglou}.}
  \bibinfo{year}{2017}\natexlab{}.
\newblock \showarticletitle{Practical attacks against graph-based clustering}.
  In \bibinfo{booktitle}{\emph{Proceedings of the 2017 ACM SIGSAC Conference on
  Computer and Communications Security}}. ACM, \bibinfo{pages}{1125--1142}.
\newblock


\bibitem[\protect\citeauthoryear{Cheng, N{\"u}hrenberg, and Ruess}{Cheng
  et~al\mbox{.}}{2017}]%
        {cheng2017maximum}
\bibfield{author}{\bibinfo{person}{Chih-Hong Cheng}, \bibinfo{person}{Georg
  N{\"u}hrenberg}, {and} \bibinfo{person}{Harald Ruess}.}
  \bibinfo{year}{2017}\natexlab{}.
\newblock \showarticletitle{Maximum resilience of artificial neural networks}.
  In \bibinfo{booktitle}{\emph{ATVA}}.
\newblock


\bibitem[\protect\citeauthoryear{Cohen, Rosenfeld, and Kolter}{Cohen
  et~al\mbox{.}}{2019}]%
        {cohen2019certified}
\bibfield{author}{\bibinfo{person}{Jeremy~M Cohen}, \bibinfo{person}{Elan
  Rosenfeld}, {and} \bibinfo{person}{J~Zico Kolter}.}
  \bibinfo{year}{2019}\natexlab{}.
\newblock \showarticletitle{Certified adversarial robustness via randomized
  smoothing}. In \bibinfo{booktitle}{\emph{ICML}}.
\newblock


\bibitem[\protect\citeauthoryear{Dai, Li, Tian, Huang, Wang, Zhu, and Song}{Dai
  et~al\mbox{.}}{2018}]%
        {dai2018adversarial}
\bibfield{author}{\bibinfo{person}{Hanjun Dai}, \bibinfo{person}{Hui Li},
  \bibinfo{person}{Tian Tian}, \bibinfo{person}{Xin Huang},
  \bibinfo{person}{Lin Wang}, \bibinfo{person}{Jun Zhu}, {and}
  \bibinfo{person}{Le Song}.} \bibinfo{year}{2018}\natexlab{}.
\newblock \showarticletitle{Adversarial attack on graph structured data}. In
  \bibinfo{booktitle}{\emph{ICML}}.
\newblock


\bibitem[\protect\citeauthoryear{Fionda and Pirro}{Fionda and Pirro}{2017}]%
        {fionda2017community}
\bibfield{author}{\bibinfo{person}{Valeria Fionda} {and}
  \bibinfo{person}{Giuseppe Pirro}.} \bibinfo{year}{2017}\natexlab{}.
\newblock \showarticletitle{Community deception or: How to stop fearing
  community detection algorithms}.
\newblock \bibinfo{journal}{\emph{IEEE Transactions on Knowledge and Data
  Engineering}} \bibinfo{volume}{30}, \bibinfo{number}{4}
  (\bibinfo{year}{2017}), \bibinfo{pages}{660--673}.
\newblock


\bibitem[\protect\citeauthoryear{Fortunato and Barthelemy}{Fortunato and
  Barthelemy}{2007}]%
        {fortunato2007resolution}
\bibfield{author}{\bibinfo{person}{Santo Fortunato} {and} \bibinfo{person}{Marc
  Barthelemy}.} \bibinfo{year}{2007}\natexlab{}.
\newblock \showarticletitle{Resolution limit in community detection}.
\newblock \bibinfo{journal}{\emph{Proceedings of the national academy of
  sciences}} \bibinfo{volume}{104}, \bibinfo{number}{1} (\bibinfo{year}{2007}),
  \bibinfo{pages}{36--41}.
\newblock


\bibitem[\protect\citeauthoryear{Gehr, Mirman, Drachsler-Cohen, Tsankov,
  Chaudhuri, and Vechev}{Gehr et~al\mbox{.}}{2018}]%
        {gehr2018ai2}
\bibfield{author}{\bibinfo{person}{Timon Gehr}, \bibinfo{person}{Matthew
  Mirman}, \bibinfo{person}{Dana Drachsler-Cohen}, \bibinfo{person}{Petar
  Tsankov}, \bibinfo{person}{Swarat Chaudhuri}, {and} \bibinfo{person}{Martin
  Vechev}.} \bibinfo{year}{2018}\natexlab{}.
\newblock \showarticletitle{Ai2: Safety and robustness certification of neural
  networks with abstract interpretation}. In \bibinfo{booktitle}{\emph{IEEE S
  \& P}}.
\newblock


\bibitem[\protect\citeauthoryear{Girvan and Newman}{Girvan and Newman}{2002}]%
        {girvan2002community}
\bibfield{author}{\bibinfo{person}{Michelle Girvan} {and}
  \bibinfo{person}{Mark~EJ Newman}.} \bibinfo{year}{2002}\natexlab{}.
\newblock \showarticletitle{Community structure in social and biological
  networks}.
\newblock \bibinfo{journal}{\emph{Proceedings of the national academy of
  sciences}} \bibinfo{volume}{99}, \bibinfo{number}{12} (\bibinfo{year}{2002}),
  \bibinfo{pages}{7821--7826}.
\newblock


\bibitem[\protect\citeauthoryear{Goodfellow, Shlens, and Szegedy}{Goodfellow
  et~al\mbox{.}}{2015}]%
        {goodfellow2014explaining}
\bibfield{author}{\bibinfo{person}{Ian~J Goodfellow}, \bibinfo{person}{Jonathon
  Shlens}, {and} \bibinfo{person}{Christian Szegedy}.}
  \bibinfo{year}{2015}\natexlab{}.
\newblock \showarticletitle{Explaining and harnessing adversarial examples}. In
  \bibinfo{booktitle}{\emph{International Conference on Learning
  Representations}}.
\newblock


\bibitem[\protect\citeauthoryear{Jia, Cao, Wang, and Gong}{Jia
  et~al\mbox{.}}{2020}]%
        {jia2020certified}
\bibfield{author}{\bibinfo{person}{Jinyuan Jia}, \bibinfo{person}{Xiaoyu Cao},
  \bibinfo{person}{Binghui Wang}, {and} \bibinfo{person}{Neil~Zhenqiang Gong}.}
  \bibinfo{year}{2020}\natexlab{}.
\newblock \showarticletitle{Certified Robustness for Top-k Predictions against
  Adversarial Perturbations via Randomized Smoothing}. In
  \bibinfo{booktitle}{\emph{International Conference on Learning
  Representations}}.
\newblock


\bibitem[\protect\citeauthoryear{Kipf and Welling}{Kipf and Welling}{2017}]%
        {kipf2017semi}
\bibfield{author}{\bibinfo{person}{Thomas~N Kipf} {and} \bibinfo{person}{Max
  Welling}.} \bibinfo{year}{2017}\natexlab{}.
\newblock \showarticletitle{Semi-supervised classification with graph
  convolutional networks}. In \bibinfo{booktitle}{\emph{ICLR}}.
\newblock


\bibitem[\protect\citeauthoryear{Lecuyer, Atlidakis, Geambasu, Hsu, and
  Jana}{Lecuyer et~al\mbox{.}}{2019}]%
        {lecuyer2018certified}
\bibfield{author}{\bibinfo{person}{Mathias Lecuyer}, \bibinfo{person}{Vaggelis
  Atlidakis}, \bibinfo{person}{Roxana Geambasu}, \bibinfo{person}{Daniel Hsu},
  {and} \bibinfo{person}{Suman Jana}.} \bibinfo{year}{2019}\natexlab{}.
\newblock \showarticletitle{Certified robustness to adversarial examples with
  differential privacy}. In \bibinfo{booktitle}{\emph{IEEE S \& P}}.
\newblock


\bibitem[\protect\citeauthoryear{Lee, Yuan, Chang, and Jaakkola}{Lee
  et~al\mbox{.}}{2019}]%
        {lee2019tight}
\bibfield{author}{\bibinfo{person}{Guang-He Lee}, \bibinfo{person}{Yang Yuan},
  \bibinfo{person}{Shiyu Chang}, {and} \bibinfo{person}{Tommi Jaakkola}.}
  \bibinfo{year}{2019}\natexlab{}.
\newblock \showarticletitle{Tight certificates of adversarial robustness for
  randomly smoothed classifiers}. In \bibinfo{booktitle}{\emph{Advances in
  Neural Information Processing Systems}}. \bibinfo{pages}{4911--4922}.
\newblock


\bibitem[\protect\citeauthoryear{Leskovec, Lang, Dasgupta, and
  Mahoney}{Leskovec et~al\mbox{.}}{2008}]%
        {leskovec2008statistical}
\bibfield{author}{\bibinfo{person}{Jure Leskovec}, \bibinfo{person}{Kevin~J
  Lang}, \bibinfo{person}{Anirban Dasgupta}, {and} \bibinfo{person}{Michael~W
  Mahoney}.} \bibinfo{year}{2008}\natexlab{}.
\newblock \showarticletitle{Statistical properties of community structure in
  large social and information networks}. In
  \bibinfo{booktitle}{\emph{Proceedings of the 17th international conference on
  World Wide Web}}. ACM, \bibinfo{pages}{695--704}.
\newblock


\bibitem[\protect\citeauthoryear{Leskovec, Lang, Dasgupta, and
  Mahoney}{Leskovec et~al\mbox{.}}{2009}]%
        {leskovec2009community}
\bibfield{author}{\bibinfo{person}{Jure Leskovec}, \bibinfo{person}{Kevin~J
  Lang}, \bibinfo{person}{Anirban Dasgupta}, {and} \bibinfo{person}{Michael~W
  Mahoney}.} \bibinfo{year}{2009}\natexlab{}.
\newblock \showarticletitle{Community structure in large networks: Natural
  cluster sizes and the absence of large well-defined clusters}.
\newblock \bibinfo{journal}{\emph{Internet Mathematics}} \bibinfo{volume}{6},
  \bibinfo{number}{1} (\bibinfo{year}{2009}), \bibinfo{pages}{29--123}.
\newblock


\bibitem[\protect\citeauthoryear{Leskovec, Lang, and Mahoney}{Leskovec
  et~al\mbox{.}}{2010}]%
        {leskovec2010empirical}
\bibfield{author}{\bibinfo{person}{Jure Leskovec}, \bibinfo{person}{Kevin~J
  Lang}, {and} \bibinfo{person}{Michael Mahoney}.}
  \bibinfo{year}{2010}\natexlab{}.
\newblock \showarticletitle{Empirical comparison of algorithms for network
  community detection}. In \bibinfo{booktitle}{\emph{Proceedings of the 19th
  international conference on World wide web}}. ACM, \bibinfo{pages}{631--640}.
\newblock


\bibitem[\protect\citeauthoryear{Li, Chen, Wang, and Carin}{Li
  et~al\mbox{.}}{2019}]%
        {li2018second}
\bibfield{author}{\bibinfo{person}{Bai Li}, \bibinfo{person}{Changyou Chen},
  \bibinfo{person}{Wenlin Wang}, {and} \bibinfo{person}{Lawrence Carin}.}
  \bibinfo{year}{2019}\natexlab{}.
\newblock \showarticletitle{Second-order adversarial attack and certifiable
  robustness}.
\newblock \bibinfo{journal}{\emph{NeurIPS}} (\bibinfo{year}{2019}).
\newblock


\bibitem[\protect\citeauthoryear{Liu, Cheng, Zhang, and Hsieh}{Liu
  et~al\mbox{.}}{2018}]%
        {liu2018towards}
\bibfield{author}{\bibinfo{person}{Xuanqing Liu}, \bibinfo{person}{Minhao
  Cheng}, \bibinfo{person}{Huan Zhang}, {and} \bibinfo{person}{Cho-Jui Hsieh}.}
  \bibinfo{year}{2018}\natexlab{}.
\newblock \showarticletitle{Towards robust neural networks via random
  self-ensemble}. In \bibinfo{booktitle}{\emph{Proceedings of the European
  Conference on Computer Vision (ECCV)}}. \bibinfo{pages}{369--385}.
\newblock


\bibitem[\protect\citeauthoryear{Madry, Makelov, Schmidt, Tsipras, and
  Vladu}{Madry et~al\mbox{.}}{2018}]%
        {madry2017towards}
\bibfield{author}{\bibinfo{person}{Aleksander Madry},
  \bibinfo{person}{Aleksandar Makelov}, \bibinfo{person}{Ludwig Schmidt},
  \bibinfo{person}{Dimitris Tsipras}, {and} \bibinfo{person}{Adrian Vladu}.}
  \bibinfo{year}{2018}\natexlab{}.
\newblock \showarticletitle{Towards deep learning models resistant to
  adversarial attacks}. In \bibinfo{booktitle}{\emph{International Conference
  on Learning Representations}}.
\newblock


\bibitem[\protect\citeauthoryear{Nagaraja}{Nagaraja}{2010}]%
        {nagaraja2010impact}
\bibfield{author}{\bibinfo{person}{Shishir Nagaraja}.}
  \bibinfo{year}{2010}\natexlab{}.
\newblock \showarticletitle{The impact of unlinkability on adversarial
  community detection: effects and countermeasures}. In
  \bibinfo{booktitle}{\emph{International Symposium on Privacy Enhancing
  Technologies Symposium}}. Springer, \bibinfo{pages}{253--272}.
\newblock


\bibitem[\protect\citeauthoryear{Newman}{Newman}{2006}]%
        {newman2006modularity}
\bibfield{author}{\bibinfo{person}{Mark~EJ Newman}.}
  \bibinfo{year}{2006}\natexlab{}.
\newblock \showarticletitle{Modularity and community structure in networks}.
\newblock \bibinfo{journal}{\emph{Proceedings of the national academy of
  sciences}} \bibinfo{volume}{103}, \bibinfo{number}{23}
  (\bibinfo{year}{2006}), \bibinfo{pages}{8577--8582}.
\newblock


\bibitem[\protect\citeauthoryear{Neyman and Pearson}{Neyman and
  Pearson}{1933}]%
        {neyman1933ix}
\bibfield{author}{\bibinfo{person}{Jerzy Neyman} {and}
  \bibinfo{person}{Egon~Sharpe Pearson}.} \bibinfo{year}{1933}\natexlab{}.
\newblock \showarticletitle{IX. On the problem of the most efficient tests of
  statistical hypotheses}.
\newblock \bibinfo{journal}{\emph{Philosophical Transactions of the Royal
  Society of London. Series A, Containing Papers of a Mathematical or Physical
  Character}} \bibinfo{volume}{231}, \bibinfo{number}{694-706}
  (\bibinfo{year}{1933}), \bibinfo{pages}{289--337}.
\newblock


\bibitem[\protect\citeauthoryear{Papernot, McDaniel, Wu, Jha, and
  Swami}{Papernot et~al\mbox{.}}{2016}]%
        {papernot2016distillation}
\bibfield{author}{\bibinfo{person}{Nicolas Papernot}, \bibinfo{person}{Patrick
  McDaniel}, \bibinfo{person}{Xi Wu}, \bibinfo{person}{Somesh Jha}, {and}
  \bibinfo{person}{Ananthram Swami}.} \bibinfo{year}{2016}\natexlab{}.
\newblock \showarticletitle{Distillation as a defense to adversarial
  perturbations against deep neural networks}. In
  \bibinfo{booktitle}{\emph{2016 IEEE Symposium on Security and Privacy (SP)}}.
  IEEE, \bibinfo{pages}{582--597}.
\newblock


\bibitem[\protect\citeauthoryear{Raghunathan, Steinhardt, and
  Liang}{Raghunathan et~al\mbox{.}}{2018}]%
        {raghunathan2018certified}
\bibfield{author}{\bibinfo{person}{Aditi Raghunathan}, \bibinfo{person}{Jacob
  Steinhardt}, {and} \bibinfo{person}{Percy Liang}.}
  \bibinfo{year}{2018}\natexlab{}.
\newblock \showarticletitle{Certified defenses against adversarial examples}.
  In \bibinfo{booktitle}{\emph{ICLR}}.
\newblock


\bibitem[\protect\citeauthoryear{Scheibler, Winterer, Wimmer, and
  Becker}{Scheibler et~al\mbox{.}}{2015}]%
        {scheibler2015towards}
\bibfield{author}{\bibinfo{person}{Karsten Scheibler}, \bibinfo{person}{Leonore
  Winterer}, \bibinfo{person}{Ralf Wimmer}, {and} \bibinfo{person}{Bernd
  Becker}.} \bibinfo{year}{2015}\natexlab{}.
\newblock \showarticletitle{Towards Verification of Artificial Neural
  Networks.}. In \bibinfo{booktitle}{\emph{MBMV}}.
\newblock


\bibitem[\protect\citeauthoryear{Szegedy, Zaremba, Sutskever, Bruna, Erhan,
  Goodfellow, and Fergus}{Szegedy et~al\mbox{.}}{2014}]%
        {Szegedy14}
\bibfield{author}{\bibinfo{person}{Christian Szegedy},
  \bibinfo{person}{Wojciech Zaremba}, \bibinfo{person}{Ilya Sutskever},
  \bibinfo{person}{Joan Bruna}, \bibinfo{person}{Dumitru Erhan},
  \bibinfo{person}{Ian Goodfellow}, {and} \bibinfo{person}{Rob Fergus}.}
  \bibinfo{year}{2014}\natexlab{}.
\newblock \showarticletitle{Intriguing properties of neural networks}. In
  \bibinfo{booktitle}{\emph{ICLR}}.
\newblock


\bibitem[\protect\citeauthoryear{Wang and Gong}{Wang and Gong}{2019}]%
        {wang2019attacking}
\bibfield{author}{\bibinfo{person}{Binghui Wang} {and}
  \bibinfo{person}{Neil~Zhenqiang Gong}.} \bibinfo{year}{2019}\natexlab{}.
\newblock \showarticletitle{Attacking Graph-based Classification via
  Manipulating the Graph Structure}. In \bibinfo{booktitle}{\emph{CCS}}.
\newblock


\bibitem[\protect\citeauthoryear{Waniek, Michalak, Wooldridge, and
  Rahwan}{Waniek et~al\mbox{.}}{2018}]%
        {waniek2018hiding}
\bibfield{author}{\bibinfo{person}{Marcin Waniek}, \bibinfo{person}{Tomasz~P
  Michalak}, \bibinfo{person}{Michael~J Wooldridge}, {and}
  \bibinfo{person}{Talal Rahwan}.} \bibinfo{year}{2018}\natexlab{}.
\newblock \showarticletitle{Hiding individuals and communities in a social
  network}.
\newblock \bibinfo{journal}{\emph{Nature Human Behaviour}} \bibinfo{volume}{2},
  \bibinfo{number}{2} (\bibinfo{year}{2018}), \bibinfo{pages}{139}.
\newblock


\bibitem[\protect\citeauthoryear{Weng, Zhang, Chen, Song, Hsieh, Boning,
  Dhillon, and Daniel}{Weng et~al\mbox{.}}{2018}]%
        {weng2018towards}
\bibfield{author}{\bibinfo{person}{Tsui-Wei Weng}, \bibinfo{person}{Huan
  Zhang}, \bibinfo{person}{Hongge Chen}, \bibinfo{person}{Zhao Song},
  \bibinfo{person}{Cho-Jui Hsieh}, \bibinfo{person}{Duane Boning},
  \bibinfo{person}{Inderjit~S Dhillon}, {and} \bibinfo{person}{Luca Daniel}.}
  \bibinfo{year}{2018}\natexlab{}.
\newblock \showarticletitle{Towards fast computation of certified robustness
  for relu networks}. In \bibinfo{booktitle}{\emph{ICML}}.
\newblock


\bibitem[\protect\citeauthoryear{Wu, Wang, Tyshetskiy, Docherty, Lu, and
  Zhu}{Wu et~al\mbox{.}}{2019}]%
        {wu2019adversarial}
\bibfield{author}{\bibinfo{person}{Huijun Wu}, \bibinfo{person}{Chen Wang},
  \bibinfo{person}{Yuriy Tyshetskiy}, \bibinfo{person}{Andrew Docherty},
  \bibinfo{person}{Kai Lu}, {and} \bibinfo{person}{Liming Zhu}.}
  \bibinfo{year}{2019}\natexlab{}.
\newblock \showarticletitle{Adversarial Examples on Graph Data: Deep Insights
  into Attack and Defense}. In \bibinfo{booktitle}{\emph{IJCAI}}.
\newblock


\bibitem[\protect\citeauthoryear{Xu, Chen, Liu, Chen, Weng, Hong, and Lin}{Xu
  et~al\mbox{.}}{2019}]%
        {xu2019topology}
\bibfield{author}{\bibinfo{person}{Kaidi Xu}, \bibinfo{person}{Hongge Chen},
  \bibinfo{person}{Sijia Liu}, \bibinfo{person}{Pin-Yu Chen},
  \bibinfo{person}{Tsui-Wei Weng}, \bibinfo{person}{Mingyi Hong}, {and}
  \bibinfo{person}{Xue Lin}.} \bibinfo{year}{2019}\natexlab{}.
\newblock \showarticletitle{Topology Attack and Defense for Graph Neural
  Networks: An Optimization Perspective}. In \bibinfo{booktitle}{\emph{IJCAI}}.
\newblock


\bibitem[\protect\citeauthoryear{Yang and Leskovec}{Yang and Leskovec}{2015}]%
        {yang2015defining}
\bibfield{author}{\bibinfo{person}{Jaewon Yang} {and} \bibinfo{person}{Jure
  Leskovec}.} \bibinfo{year}{2015}\natexlab{}.
\newblock \showarticletitle{Defining and evaluating network communities based
  on ground-truth}.
\newblock \bibinfo{journal}{\emph{Knowledge and Information Systems}}
  \bibinfo{volume}{42}, \bibinfo{number}{1} (\bibinfo{year}{2015}),
  \bibinfo{pages}{181--213}.
\newblock


\bibitem[\protect\citeauthoryear{Zhu, Zhang, Cui, and Zhu}{Zhu
  et~al\mbox{.}}{2019}]%
        {zhu2019robust}
\bibfield{author}{\bibinfo{person}{Dingyuan Zhu}, \bibinfo{person}{Ziwei
  Zhang}, \bibinfo{person}{Peng Cui}, {and} \bibinfo{person}{Wenwu Zhu}.}
  \bibinfo{year}{2019}\natexlab{}.
\newblock \showarticletitle{Robust Graph Convolutional Networks Against
  Adversarial Attacks}. In \bibinfo{booktitle}{\emph{KDD}}.
\newblock


\bibitem[\protect\citeauthoryear{Z{\"u}gner, Akbarnejad, and
  G{\"u}nnemann}{Z{\"u}gner et~al\mbox{.}}{2018}]%
        {zugner2018adversarial}
\bibfield{author}{\bibinfo{person}{Daniel Z{\"u}gner}, \bibinfo{person}{Amir
  Akbarnejad}, {and} \bibinfo{person}{Stephan G{\"u}nnemann}.}
  \bibinfo{year}{2018}\natexlab{}.
\newblock \showarticletitle{Adversarial attacks on neural networks for graph
  data}. In \bibinfo{booktitle}{\emph{KDD}}.
\newblock


\bibitem[\protect\citeauthoryear{Z{\"u}gner and G{\"u}nnemann}{Z{\"u}gner and
  G{\"u}nnemann}{2019a}]%
        {zugner2019adversarial}
\bibfield{author}{\bibinfo{person}{Daniel Z{\"u}gner} {and}
  \bibinfo{person}{Stephan G{\"u}nnemann}.} \bibinfo{year}{2019}\natexlab{a}.
\newblock \showarticletitle{Adversarial attacks on graph neural networks via
  meta learning}.
\newblock \bibinfo{journal}{\emph{ICLR}} (\bibinfo{year}{2019}).
\newblock


\bibitem[\protect\citeauthoryear{Z{\"u}gner and G{\"u}nnemann}{Z{\"u}gner and
  G{\"u}nnemann}{2019b}]%
        {Zugner2019Certifiable}
\bibfield{author}{\bibinfo{person}{Daniel Z{\"u}gner} {and}
  \bibinfo{person}{Stephan G{\"u}nnemann}.} \bibinfo{year}{2019}\natexlab{b}.
\newblock \showarticletitle{Certifiable Robustness and Robust Training for
  Graph Convolutional Networks}. In \bibinfo{booktitle}{\emph{KDD}}.
\newblock


\end{thebibliography}

\appendix
\section{A Variant of the Neyman-Pearson Lemma}
\label{proof_of_lemma_neyman_pearson}

We prove a variant of the Neyman-Pearson Lemma~\cite{neyman1933ix} that is used to derive our certified perturbation size.

\begin{restatable}{lemma}{NPL}
\label{neyman-pearson-variant}
Assume $X$ and $Y$ are two random variables in the discrete space $\{0,1\}^{n}$ with probability distributions $\textrm{Pr}(X)$ and $\textrm{Pr}(Y)$, respectively. Let $\psi: \{0,1\}^{n} \rightarrow \{0, 1\}$ be a random or deterministic function. 
Let $T_1 = \{ \mathbf{z} \in \{0, 1\}^n: \frac{\textrm{Pr}(X=\mathbf{z})}{\textrm{Pr}(Y=\mathbf{z})} > t \}$ 
and 
$T_2 = \{ \mathbf{z} \in \{0, 1\}^n: \frac{\textrm{Pr}(X=\mathbf{z})}{\textrm{Pr}(Y=\mathbf{z})} = t \}$ for some $t>0$.
Assume $T_3 \subseteq T_2$ and $T = T_1 \bigcup T_3$. If $\textrm{Pr}(\psi(X)=1) \geq \textrm{Pr}(X \in T)$, then $\textrm{Pr}(\psi(Y)=1) \geq \textrm{Pr}(Y \in T)$.
\end{restatable}
\begin{proof}

In general, we assume that $\psi$ is a random function. We use $\psi(0|\mathbf{z})$ to denote the probability that $\psi(\mathbf{z})=1$ and $\psi(0|\mathbf{z})$ to denote the probability that $\psi(\mathbf{z})=0$, respectively. We use $T^c$ to represent the complement of $T$, i.e., $T^c=\{0, 1\}^n\setminus T$.
\begin{align}
& \textrm{Pr}(\psi(Y)=1) - \textrm{Pr}(Y \in T)  \\
& = \sum_{\mathbf{z} \in \{0,1\}^{n}} \psi(1|\mathbf{z}) \textrm{Pr}(Y = \mathbf{z}) - \sum_{\mathbf{z} \in T} \textrm{Pr}(Y = \mathbf{z}) \\
& = \Big( \sum_{\mathbf{z} \in T^{c}} \psi(1|\mathbf{z}) \textrm{Pr}(Y = \mathbf{z}) + \sum_{\mathbf{z} \in T} \psi(1|\mathbf{z}) \textrm{Pr}(Y = \mathbf{z}) \Big)  \\ 
&   \qquad - \Big( \sum_{\mathbf{z} \in T} \psi(1|\mathbf{z}) \textrm{Pr}(Y = \mathbf{z}) + \sum_{\mathbf{z} \in T} \psi(0|\mathbf{z}) \textrm{Pr}(Y = \mathbf{z}) \Big) \\
\label{neyman_pearson_prove_inequality_1}
& = \sum_{\mathbf{z} \in T^{c}} \psi(1|\mathbf{z}) \textrm{Pr}(Y = \mathbf{z}) - \sum_{\mathbf{z} \in T} \psi(0|\mathbf{z}) \textrm{Pr}(Y = \mathbf{z}) \\
\label{neyman_pearson_prove_inequality_2}
& \geq \frac{1}{t} \Big( \sum_{\mathbf{z} \in T^{c}} \psi(1|\mathbf{z}) \textrm{Pr}(X = \mathbf{z}) - \sum_{\mathbf{z} \in T} \psi(0|\mathbf{z}) \textrm{Pr}(X = \mathbf{z}) \Big) \\
& = \frac{1}{t} \Bigg( \Big( \sum_{\mathbf{z} \in T^{c}} \psi(1|\mathbf{z}) \textrm{Pr}(X = \mathbf{z}) + \sum_{\mathbf{z} \in T} \psi(1|\mathbf{z}) \textrm{Pr}(X = \mathbf{z}) \Big)  \\ 
& \qquad - \Big( \sum_{\mathbf{z} \in T} \psi(1|\mathbf{z}) \textrm{Pr}(X = \mathbf{z}) + \sum_{\mathbf{z} \in T} \psi(0|\mathbf{z}) \textrm{Pr}(X = \mathbf{z}) \Big) \Bigg) \\ 
& = \frac{1}{t} \Big( \sum_{\mathbf{z} \in \{0,1\}^{n}} \psi(1|\mathbf{z}) \textrm{Pr}(X = \mathbf{z}) - \sum_{\mathbf{z} \in T} \textrm{Pr}(X = \mathbf{z}) \Big) \\
\label{neyman_pearson_prove_inequality_3}
& = \frac{1}{t} \Big( \textrm{Pr}(\psi(X)=1) - \textrm{Pr}(X \in T) \Big) \\
\label{neyman_pearson_prove_inequality_4}
& \geq 0.
\end{align}
We obtain Equation~\ref{neyman_pearson_prove_inequality_2} from~\ref{neyman_pearson_prove_inequality_1} because for any $\mathbf{z} \in T$, we have $\frac{\textrm{Pr}(X=\mathbf{z})}{\textrm{Pr}(Y=\mathbf{z})} \geq t$, and for any $\mathbf{z} \in T^c$,  we have $\frac{\textrm{Pr}(X=\mathbf{z})}{\textrm{Pr}(Y=\mathbf{z})} \leq t$. We have Equation~\ref{neyman_pearson_prove_inequality_4} from~\ref{neyman_pearson_prove_inequality_3} because $\textrm{Pr}(\psi(X)=1) \geq \textrm{Pr}(X \in T)$. 
\end{proof}

\section{Proof of Theorem~\ref{CertifiedPerturbationSize}}
\label{proof_of_theorem_certifiedperturbationsize}
We define the following two random variables: 
\begin{align}
X &= \mathbf{x} \oplus \bm{\epsilon} \\
Y &= \mathbf{x} \oplus \bm{\delta} \oplus \bm{\epsilon}, 
\end{align}
where the random variables $X$ and $Y$ denote random samples by adding discrete noise to the graph-structure binary vector $\mathbf{x}$ and its perturbed version $\mathbf{x}+\bm{\delta}$, respectively. 
Our goal is to find the maximum perturbation size $||\bm{\delta}||_0$ such that the following holds:  
\begin{align}
    \text{Pr}(f(Y)=y)> 0.5
\end{align}
First, we define region $\mathcal{Q}$ such that $\text{Pr}(X \in \mathcal{Q}) = \underline{p}$. Specifically, we gradually add the regions $\mathcal{H}_1, \mathcal{H}_2, \cdots, \mathcal{H}_{2n+1}$ to $\mathcal{Q}$ until  $\text{Pr}(X \in \mathcal{Q}) = \underline{p}$. In particular, we define $\mu$ as follows:
\begin{align}
\mu = \argmin_{\mu^{\prime} \in \{1, 2, \cdots, 2n+1 \}} \mu^{\prime}, \ s.t.\ \sum_{i=1}^{\mu^{\prime}} \textrm{Pr}(X \in \mathcal{H}_i) \geq \underline{p}
\end{align}
Moreover, we define $\underline{\mathcal{H}_{\mu}}$ as any subregion of $\mathcal{H}_{\mu}$ such that: 
\begin{align}
\text{Pr}(X\in \underline{\mathcal{H}_{\mu}})= \underline{p} - \sum_{i=1}^{\mu-1} \textrm{Pr}(X \in \mathcal{H}_i) 
\end{align}
Then, we define the region $\mathcal{Q}$ as follows:
\begin{align}
  \mathcal{Q} = \bigcup_{i=1}^{\mu-1} \mathcal{H}_i \cup \underline{\mathcal{H}_{\mu}}
\end{align} 
Based on the conditions in Equation~\ref{main_theorem_condition_label}, we have:
\begin{align}
  \text{Pr}(f(X)=y)\geq \underline{p}= \text{Pr}(X\in\mathcal{Q})
\end{align}
We define a function $\psi(\mathbf{z})=\mathbb{I}(f(\mathbf{z})=y)$. Then, we have: 
\begin{align}
\text{Pr}(\psi(X)=y) = \text{Pr}(f(X)=y)\geq \text{Pr}(X\in\mathcal{Q}).
\end{align}
Moreover,  we have $\frac{\text{Pr}(X=\mathbf{z})}{\text{Pr}(Y=\mathbf{z})} > h_{\mu}$ if and only if $\mathbf{z} \in \bigcup_{j=1}^{\mu-1} \mathcal{H}_j$, 
 and $\frac{\text{Pr}(X=\mathbf{z})}{\text{Pr}(Y=\mathbf{z})} = h_{\mu}$ for any $\mathbf{z} \in \underline{\mathcal{H}_{\mu}}$, where $h_{\mu}$ is the probability density ratio in the region $\mathcal{H}_{\mu}$. Note that $\mathcal{Q} = \bigcup_{i=1}^{\mu-1} \mathcal{H}_i \cup \underline{\mathcal{H}_{\mu}}$. Therefore, according to Lemma~\ref{neyman-pearson-variant}, we have:
\begin{align}
  \text{Pr}(f(Y)=y) \geq \text{Pr}(Y\in \mathcal{Q}). 
\end{align}
Therefore, to reach our goal, it is sufficient to have: 
\begin{align}
 & \text{Pr}(Y\in \mathcal{Q})>0.5  \Longleftrightarrow\\
 & \text{Pr}(Y\in \mathcal{Q}) \\
 = &\text{Pr}(Y \in \bigcup_{i=1}^{\mu-1} \mathcal{H}_i \cup \underline{\mathcal{H}_{\mu}})  \\
 =& \sum_{i=1}^{\mu-1} \text{Pr} (Y \in \mathcal{H}_i) + (\underline{p} - \sum_{i=1}^{\mu-1} \text{Pr} (X \in \mathcal{H}_i)) /  h_{\mu} \\
>& 0.5, 
\end{align}
where $h_{\mu}$ is the probability density ratio in the region $\mathcal{H}_{\mu}$. Our certified perturbation size $L$ is the largest perturbation size $||\bm{\delta}||_0$ that makes the above inequality hold. 

\section{Proof of Theorem~\ref{tighnessbound}}
\label{proof_of_tight_of_certified_perturbation_size}
Our idea is to construct a base function $f^*$ consistent with the conditions in Equation~\ref{main_theorem_condition_label}, but the smoothed function is not guaranteed to predict $y$. 
Let region $\mathcal{Q}$ be defined as in the proof of Theorem~\ref{CertifiedPerturbationSize} We denote $\mathcal{Q}^c$ as the complement of $\mathcal{Q}$, i.e., $\mathcal{Q}^c=\{0,1\}^n \setminus \mathcal{Q}$. 
Given region $\mathcal{Q}$, we construct the following function:
\begin{align}
    f^{*}(\mathbf{z})=
    \begin{cases}
    y  & \text{if }\mathbf{z}\in\mathcal{Q} \\
    1-y & \text{if }\mathbf{z}\in\mathcal{Q}^c
    \end{cases}
\end{align}
By construction, we have $\text{Pr}(f^{*}(X)=y)=\underline{p}$ and  $\text{Pr}(f^{*}(X)=1-y)=1-\underline{p}$, which are consistent with Equation~\ref{main_theorem_condition_label}. 
From our proof of Theorem~\ref{CertifiedPerturbationSize}, we know that when $||\bm{\delta}||_0 > L$, we have:
\begin{align}
    \text{Pr}(Y\in\mathcal{Q}) \leq 0.5,
\end{align}
Moreover, based on the definition of $f^{*}$, we have:
\begin{align}
\text{Pr}(f^{*}(Y)=y) \leq 
\text{Pr}(f^{*}(Y)=1-y)
\end{align}
Therefore, we have either $g(\mathbf{x} \oplus \bm{\delta})\neq y$ or there exists ties when $||\bm{\delta}||_0 > L$.

\section{Computing $\text{Pr}(\mathbf{x}\oplus\epsilon \in \mathcal{H}(e))$ and $\text{Pr}(\mathbf{x} \oplus\delta \oplus \epsilon \in \mathcal{H}(e))$}
\label{compute_algorithm_part}
We first define the following regions:
\begin{align}
\mathcal{H}(a,b) & = \{\mathbf{z} \in \{0, 1\}^n: ||\mathbf{z}-\mathbf{x}||_0=a \text{ and } \\ 
& ||\mathbf{z}-\mathbf{x} \oplus \bm{\delta}||_0 =  b \},
\end{align}
for $a,b \in \{0, 1, \cdots, n \}$. Intuitively, $\mathcal{H}(a,b)$ includes the binary vectors that are $a$ bits different from $\mathbf{x}$ and $b$ bits different from 
$\mathbf{x} \oplus \bm{\delta}$.  
 Next, we compute the size of the region $\mathcal{H}(a,b)$ when $||\bm{\delta}||_0=l$. Without loss of generality, we assume $\mathbf{x} = [0, 0, \cdots, 0]$ as a zero vector and $\mathbf{x} \oplus \bm{\delta} = [1, 1, \cdots, 1, 0, 0, \cdots, 0]$, where the first $l$ entries are 1 and the remaining $n - l$ entries are 0. We construct a binary vector $\mathbf{z}\in \mathcal{H}(a,b)$. Specifically,  suppose we flip $i$ zeros in the last $n-l$ zeros in both $\mathbf{x}$ and $\mathbf{x} \oplus \bm{\delta}$. Then, we flip $a-i$ of the first $l$ bits of  $\mathbf{x}$ and flip the rest $l-a+i$ bits of the first $l$ bits of  $\mathbf{x} \oplus \bm{\delta}$. In order to have $\mathbf{z}\in \mathcal{H}(a,b)$, we need  $l-a+i + i = b$, i.e., $i=(a+b-l)/2$. Therefore, we have the size  $|\mathcal{H}(a,b)|$ of the region $\mathcal{H}(a,b)$ as follows:
\begin{align}
  |\mathcal{H}(a,b)|  = 
    \begin{cases}
      0,  &\textrm{ if } (a+b-l) \textrm{ mod } 2 \neq 0, \\
      0,  &\textrm{ if } a+b<l, \\
  {n-l \choose \frac{a+b-l}{2}} {l \choose \frac{a-b+l}{2}}, &\textrm{ otherwise } 
  \end{cases}
\end{align}
Moreover, for each $\mathbf{z} \in \mathcal{H}(a,b)$, we have $\text{Pr}(\mathbf{x} \oplus \bm{\epsilon} = \mathbf{z}) = \beta^{n-a} (1-\beta)^{a}$  and $\text{Pr}(\mathbf{x} \oplus \bm{\delta} \oplus \bm{\epsilon} = \mathbf{z}) = \beta^{n-b} (1-\beta)^b$. Therefore, we have:
\begin{align}
  \textrm{Pr}(\mathbf{x} \oplus \bm{\epsilon} \in \mathcal{H}(a,b)) 
  & = \beta^{n-a} (1-\beta)^a \cdot |\mathcal{H}(a,b)| \\ 
  \textrm{Pr}(\mathbf{x} \oplus \bm{\delta} \oplus \bm{\epsilon} \in \mathcal{H}(a,b)) 
  & = \beta^{n-b} (1-\beta)^b \cdot |\mathcal{H}(a,b)| 
\end{align}

Note that $\mathcal{H}(e)=\cup_{b-a=e} \mathcal{H}(a, b)$. 
 Therefore, we have: 
\begin{align}
  & \text{Pr}(\mathbf{x} \oplus \bm{\epsilon} \in \mathcal{H}(e)) \\
  =& \text{Pr}(\mathbf{x} \oplus \bm{\epsilon} \in \cup_{b-a=e} \mathcal{H}(a, b)) \\
  =  & \text{Pr}(\mathbf{x} \oplus \bm{\epsilon} \in \cup_{i=\max(0,e)}^{\min(n,n+e)} \mathcal{H}(i-e,i))  \\ 
  =  & \sum_{i=\max(0,e)}^{\min(n,n+e)} \text{Pr}(\mathbf{x} \oplus \bm{\epsilon} \in \mathcal{H}(i-e,i))  \\ 
  =  & \sum_{i=\max(0,e)}^{\min(n,n+e)} \beta^{n-(i-e)} (1-\beta)^{i-e} \cdot |\mathcal{H}(i-e,i)| \label{probX} \\
    =  & \sum_{i=\max(0,e)}^{\min(n,n+e)} \beta^{n-(i-e)} (1-\beta)^{i-e} \cdot \theta(e,i),
\end{align}

where $\theta(e,i)=|\mathcal{H}(i-e,i)|$.

Similarly, we have
\begin{align}
& \text{Pr}(\mathbf{x} \oplus \bm{\delta} \oplus \bm{\epsilon} \in \mathcal{H}(e)) \\
 =& \sum_{i=\max(0,e)}^{\min(n,n+e)} \beta^{n-i} (1-\beta)^{i} \cdot \theta(e,i).
\end{align}

\section{Proof of Proposition~\ref{probability_guarantees}}
\label{proof_of_probability_guarantee}

Based on the Clopper-Pearson  method, we know the probability that the following inequality holds is at least $1-\alpha$ over the randomness in sampling the noise: 
\begin{align}
  \text{Pr}(f(\mathbf{x}+\bm{\epsilon})=\hat{y})  \ge \underline{p}
\end{align}
Therefore, with the estimated probability lower bound $\underline{p}$, we can invoke Theorem \ref{CertifiedPerturbationSize} to obtain the robustness guarantee if $\underline{p}> 0.5$. Note that otherwise our algorithm \textsc{Certify}  abstains.

\end{document}